\documentclass[12pt]{amsart}
\usepackage{amsmath}
\usepackage{amsxtra}
\usepackage{amscd}
\usepackage{amsthm}
\usepackage{amsfonts}
\usepackage{amssymb}
\usepackage{eucal}
\usepackage{epsfig}
\usepackage{graphics}
\def\({\left(}
\def\){\right)}

\newcommand{\xib}{\mbox{\boldmath$\xi$}}

\newcommand{\taub}{\mbox{\boldmath$\tau$}}

\newcommand{\ket}[1]{{| #1 \rangle}}      

\newcommand{\cb}{\mathbf{c}}
\newcommand{\bb}{\mathbf{b}}
\newcommand{\ab}{\mathbf{a}}
\newcommand{\tb}{\mathbf{t}}
\newcommand{\fb}{\mathbf{f}}
\newcommand{\ub}{\mathbf{u}}
\newcommand{\qb}{\mathbf{q}}
\newcommand{\gb}{\mathbf{g}}

\newcommand{\nn}{\nonumber}
\newcommand{\bea}{\begin{eqnarray}}
\newcommand{\ena}{\end{eqnarray}}
\def\bel{\begin{eqnarray}}
\def\enl{\end{eqnarray}}
\newcommand{\be}{\begin{eqnarray*}}
\newcommand{\en}{\end{eqnarray*}}
\newcommand{\ba}{\begin{array}}
\newcommand{\ea}{\end{array}}

\newcommand{\R}{{\mathbb R}}
\newcommand{\T}{{\mathbb T}}

\newcommand{\C}{{\mathbb C}}
\newcommand{\Z}{{\mathbb Z}}

\newcommand{\slt}{\mathfrak{sl}_2}
\newcommand{\slth}{\widehat{\mathfrak{sl}}_2}
\newcommand{\res}{{\rm res}}

\newcommand{\Tr}{{\rm Tr}}

\newcommand{\End}{\mathop{\rm End}}

\newenvironment{tenumerate}{
  \begin{enumerate}
  
  }{\end{enumerate}}
\newcommand{\bi}{\begin{tenumerate}}
\newcommand{\ei}{\end{tenumerate}}
\newcommand{\isoto}[1][]%
{{\mathop{\buildrel{\sim}\over\longrightarrow}\limits_{#1}}}

\def\tto{\underset{\nu\to 0}{\rightarrow}}
\def\[{\left[}
\def\]{\right]}
\newcommand{\la}{\lambda}

\newcommand{\al}{\alpha}

\newcommand{\s}{\sigma}
\newcommand{\z}{\zeta}

\numberwithin{equation}{section}
\newtheorem{thm}{Theorem}[section]

\newtheorem{lem}[thm]{Lemma}

\newtheorem{cor}[thm]{Corollary}

\def\rb{\mathbf{r}}
\def\xb{\mathbf{x}}
\def\J{\mathbb{J}}
\newcommand{\osym}{\omega_{sym}}

\newcommand{\bS}{\mathbb{S}}

\newcommand{\ao}{\mathbf{a}}

\newcommand{\kb}{\mathbf{k}}
\def\half{\textstyle{\frac  1 2}}
\newcommand{\bk}{\mathbf{k}}
\newcommand{\bl}{\mathbf{l}}
\newcommand{\bT}{{\mathbb T}}

\newcommand{\bu}{\mathbf{u}}
\newcommand{\bp}{\mathfrak{b}^+}
\newcommand{\bbA}{\mathbb{A}}
\def\bo{\mathbf{1}}

\def\bn{\mathbf{n}}
\def\bno{\mathbf{n-1}}
\def\bj{\mathbf{j}}
\def\bi{\mathbf{i}}
\def\bp{\mathbf{p}}
\begin{document}

\begin{title}[Grassmann Structure in XXZ Model]
{Hidden Grassmann Structure in the XXZ Model III:
Introducing Matsubara direction
}
\end{title}
\date{\today}
\author{M.~Jimbo, T.~Miwa and  F.~Smirnov}
\address{MJ: Graduate School of Mathematical Sciences, The
University of Tokyo, Tokyo 153-8914, Japan;
Institute for the Physics and Mathematics of the Universe, 
Kashiwa, Chiba 277-8582, Japan}\email{jimbomic@ms.u-tokyo.ac.jp}
\address{TM: Department of 
Mathematics, Graduate School of Science,
Kyoto University, Kyoto 606-8502, 
Japan}\email{tetsuji@math.kyoto-u.ac.jp}
\address{FS\footnote{Membre du CNRS}: Laboratoire de Physique Th{\'e}orique et
Hautes Energies, Universit{\'e} Pierre et Marie Curie,
Tour 16 1$^{\rm er}$ {\'e}tage, 4 Place Jussieu
75252 Paris Cedex 05, France}\email{smirnov@lpthe.jussieu.fr}

\begin{abstract}
We address the problem of computing temperature correlation functions 
of the XXZ chain, within the approach developed in our previous works. 
In this paper we calculate 
the expected values of a fermionic basis of quasi-local operators, 
in the infinite volume limit 
while keeping the Matsubara (or Trotter) 
direction finite.
The result is expressed in terms of two basic quantities: a ratio 
$\rho(\z)$ of transfer matrix eigenvalues, and a nearest neighbour correlator
$\omega(\z,\xi)$. We explain that the latter is interpreted as 
the canonical second kind differential in the theory of 
deformed Abelian integrals. 
\end{abstract}

\maketitle

\bigskip


\section{Introduction}\label{Intro}

The present article is a continuation of 
the paper \cite{HGSII}, which was written almost
a year ago and was dedicated to the memory of Alyosha Zamolodchikov. 
It so happens that the topic we discuss this time is 
not too far from a domain in which he made giant footsteps. 
So, life goes on, but there stays
a painful sorrow caused by his early death. 

Consider the XXZ spin chain with the Hamiltonian 
\begin{eqnarray}
H=\textstyle{\frac{1}{2}}\sum\limits_{k=-\infty}^{\infty}
\left( 
\sigma_{k}^1\sigma_{k+1}^1+
\sigma_{k}^2\sigma_{k+1}^2+
\Delta\sigma_{k}^3\sigma_{k+1}^3
\right), \quad \Delta =\half(q+q^{-1})\,,
\label{Ham}
\end{eqnarray}
where $\sigma^a \, (a=1,2,3)$ are the Pauli matrices. 
To avoid technicalities, 
in this Introduction let us accept \eqref{Ham} as 
a formal object acting on 
$\frak{H}_\mathrm{S} 
=\bigotimes\limits_{j=-\infty}^{\infty}\mathbb{C}^2$.  
We shall touch upon the limit 
from a finite chain in the body of the text.   
In the papers \cite{HGS}, \cite{HGSII}, we studied the vacuum
expectation values (VEVs)
\begin{align}
\langle q^{2\al S(0)}\mathcal{O}\rangle_{{XXZ}}
=\frac{\langle\text{vac}|q^{2\al S(0)}
\mathcal{O}|\text{vac}\rangle}{\langle\text{vac}|
q^{2\al S(0)}|\text{vac}\rangle}\,.
\label{exp}
\end{align}
Here $|\text{vac}\rangle$ denotes the ground state eigenvector,  
$S(k)=\textstyle{\frac 1 2} \sum_{j=-\infty}^k\sigma ^3_j$,
and $\mathcal{O}$ is a local operator. 
We have obtained a description of \eqref{exp} 
in terms of fermionic operators. 
For that purpose, it was essential to consider 
operators of the form $q^{2\al S(0)}\mathcal{O}$,  
which we call quasi-local operators with tail $\al$. 

An important generalisation
of our results was proposed by Boos, G{\"o}hmann, Kl{\"u}mper and Suzuki 
\cite{BGKS}. 
They gave evidences that our fermionic description
works equally well 
in the presence of
a finite temperature and a non-zero magnetic field:
\begin{align}
\langle q^{2\al S(0)}\mathcal{O}\rangle_{{XXZ},\ \beta ,h}\ =\ 
\frac{\Tr _{\mathrm{S}}\Bigl(e^{-\beta H+hS}q^{2\al S(0)}\mathcal{O}\Bigr)}
{\Tr _{\mathrm{S}}\Bigl(e^{-\beta H+hS}q^{2\al S(0)}\Bigr)}\,, 
\label{TEV}
\end{align}
where $\Tr_\mathrm{S}$ stands for the trace on $\frak{H}_\mathrm{S}$. 
For $\beta \to \infty$ and $h=0$, 
the expectation value (\ref{TEV}) reduces to (\ref{exp}). 
For us this was quite an exciting development,  
because it shows that the
fermionic structure is not 
a peculiarity of VEVs, but is rather a 
reflection of a symmetry hidden deep in the model. 
It should be said that in the paper \cite{BGKS} the
expectation values (\ref{TEV}) were not considered in full
generality. The formula expressing them
in terms of fermionic operators was formulated as 
a conjecture, which was 
checked in some particular cases but was left unproved. 

The first question which we asked ourselves was, why not to
add other local integrals of motion to $-\beta H+hS$ in
(\ref{TEV}). 
The physical meaning of such a generalisation 
is obscure, but it should be possible for integrable models. 
This question,  
together with an intuition coming from the papers \cite{Klumper},
\cite{Suzuki}, led to the following generalisation of (\ref{TEV}).
Along with the space $\frak{H}_\mathrm{S}$, 
consider the Matsubara space $\frak{H}_{\mathbf{M}}$, 
\begin{align}
\frak{H}_\mathbf{M}=\mathbb{C}^{2s_\mathbf{1}+1}\otimes
\cdots \otimes \mathbb{C}^{2s_\mathbf{n}+1}\,,
\label{HM}
\end{align}
with an arbitrary spin $s_{\mathbf{m}}$ and a 
spectral parameter $\tau_\mathbf{m}$ attached to each component. 
The generalisation of (\ref{TEV}) is given
by the following linear functional
\begin{align}
Z^{\kappa}\Bigl\{q^{2\al S(0)}\mathcal{O}
\Bigr\}=
\frac{\Tr _{\mathrm{S}}\Tr _{\mathbf{M}}\Bigl(T_{\mathrm{S},\mathbf{M}}q^{2\kappa S+2\al S(0)}\mathcal{O}\Bigr)}
{\Tr _{\mathrm{S}}\Tr _\mathbf{M}\Bigl(T_{\mathrm{S},\mathbf{M}}q^{2\kappa S+2\al S(0)}\Bigr)}
\,.
\label{Z}
\end{align}
Here $T_{S,\mathbf{M}}$ denotes the 
monodromy matrix associated with 
$\mathfrak{H}_\mathrm{S}\otimes\mathfrak{H}_{\mathbf{M}}$
(see \eqref{monodromy-matrix}). 

The idea behind the generalisation (\ref{Z}) is simple:
for whichever spins and $\tau _\mathbf{m}$ that
we put in the Matusbara direction, 
$\Tr _\mathbf{M}\bigl(T_{\mathrm{S},\mathbf{M}}\bigr)$ commutes with $H_{XXZ}$. 
One expects that using cleverly 
this arbitrariness in the definition of 
$\frak{H}_{\mathbf{M}}$, 
it should be possible to reproduce any function of local 
integrals of motion under the trace.  
In particular, in order to reproduce (\ref{TEV}) from
(\ref{Z}),  one has to take special inhomogeneities and
then to consider the limit $\mathbf{n}\to\infty$. 
This point is explained in detail 
in \cite{Klumper}, \cite{Suzuki}. 
In the present paper we compute $Z^{\kappa}$ for  finite 
$\mathbf{n}$, leaving the discussion of the limit for future
publication. 
We would like to emphasise, however, that
this limit is not complicated. 
For finite $\mathbf{n}$, $Z^{\kappa}$ will be
expressed in terms of only two functions,  
$\rho (\z)$, $\omega (\z,\xi)$ 
(see \eqref{detform} below) and one needs only to take the limit of them. 
Let us explain all that in some more details,
starting from our fermionic operators.

For the moment we forget about the 
Matsubara direction, and concentrate
on the description of the operators 
acting on $\mathfrak{H}_\mathrm{S}$.
The logic of our papers \cite{HGS}, \cite{HGSII} is close to that of CFT:
we describe the space of quasi-local operators as a module
created from the primary field $q^{2\al S(0)}$ by creation
operators. 
We recall below the main features of the construction in 
\cite{HGSII}.

We say that $X=q^{2\al S(0)}\mathcal{O}$ is a 
quasi-local operator with tail $\al$ 
if it stabilises outside some finite interval
of the infinite chain: to $q^{\al \sigma ^3_j}$ 
on the left and to $I_j$ on the right. 
The minimal interval with this 
property is called the support of $X$.
The spin of $X$ is the eigenvalue of $\mathbb{S}(\cdot )=[S,\cdot]$ 
where $S=S(\infty)$ is the total spin operator. 
We denote by $\mathcal{W}_{\al}$ the 
space of quasi-local operators 	with tail $\alpha$, and by 
$\mathcal{W}_{\al,s}$ its subspace 
of operators of spin $s\in \mathbb{Z}$. 
Consider the space
$$
\mathcal{W}^{(\al)}
=\bigoplus\limits _{s=-\infty}^{\infty}\mathcal{W}_{\al-s,s}
\,.
$$
On this space we defined the creation 
operators $\tb ^*(\z)$, $\bb^*(\z)$, $\cb^*(\z)$ 
and annihilation operators $\bb(\z)$, $\cb(\z)$. 
These are one-parameter families of operators of the form
\begin{align}
&\tb ^*(\z)=\sum\limits _{p=1}^{\infty}(\z^2-1)^{p-1}\tb_p\,,
\nn\\
&
\ \bb ^*(\z)=\z^{\al+2}\sum\limits _{p=1}^{\infty}(\z^2-1)^{p-1}\bb^*_p\,,
\ \cb ^*(\z)=\z^{-\al-2}\sum\limits _{p=1}^{\infty}(\z^2-1)^{p-1}\cb^*_p\,,\nn
\\
&\bb (\z)=\z^{-\al}\sum\limits _{p=0}^{\infty}(\z^2-1)^{-p}\bb_p\,,
\ \cb (\z)=\z^{\al}\sum\limits _{p=0}^{\infty}(\z^2-1)^{-p}\cb_p\,.\nn
\end{align}
The operator $\tb^*(\z)$ is in the center of our algebra 
of creation-annihilation operators,  
\begin{align}
&[\tb^*(\z_1),\tb^*(\z _2)]=[\tb^*(\z_1),\cb^*(\z _2)]=
[\tb^*(\z_1),\bb^*(\z _2)]=0,
\nn\\
&[\tb^*(\z_1),\cb(\z _2)]=
[\tb^*(\z_1),\bb(\z _2)]=0\,.
\nn
\end{align}
The rest of the operators $\bb$, $\cb$, $\bb ^*$, $\cb ^*$ 
are fermionic. The only non-vanishing anti-commutators are 
\begin{align}
&[\bb (\z _1 ),\bb ^*(\z _2)]_+
=-\psi (\z_2/\z_1,\al)\,,\quad
[\cb (\z _1 ),\cb ^*(\z _2)]_+=
\psi (\z_1/\z_2,\al)\,,
\nn
\end{align}
where
\begin{align}
\psi (\z,\al)=\z ^{\al}\frac {\z ^2+1}{2(\z^2-1)}\,.
\label{psi}
\end{align}

Each Fourier mode has the block structure
\begin{align}
&\tb_p ^*\ :\ \mathcal{W}_{\al-s,s}\ \to\ \mathcal{W}_{\al-s,s}\label{blocks}\\
&\bb_p ^*,\cb_p\ :\ \mathcal{W}_{\al-s+1,s-1}\ \to\ \mathcal{W}_{\al-s,s}\,,
\quad \cb_p ^*,\bb_p\ :\ \mathcal{W}_{\al-s-1,s+1}\ \to\ \mathcal{W}_{\al-s,s}\,.\nn
\end{align}
Among them, $\taub=\tb^*_1/2$ plays a special role. 
It is the right shift by one site along the chain. 
Consider the set of operators
\begin{align}
\taub^{m}\tb^*_{p_1}\cdots \tb^*_{p_j}
\bb^*_{q_1}\cdots \bb^*_{q_k}
\cb^*_{r_1}\cdots \cb^*_{r_k}\Bigl(q^{2\al S(0)}\Bigr),
\label{basis}
\end{align}
where $m\in \Z$, $j,k\in\Z_{\ge0}$, 
$p_1\ge\cdots\ge p_j\ge2$, 
$q_1>\cdots>q_k\ge 1$ and 
$r_1>\cdots>r_k\ge 1$. 
It can be shown that \eqref{basis}
constitutes a basis of $\mathcal{W}_{\al,0}$
(we postpone the proof to other publication).

Now we start to consider the spaces
$\mathfrak{H}_\mathrm{S}$ and $\mathfrak{H}_\mathbf{M}$ 
together. We shall prove that
\begin{align}
&Z^{\kappa}\bigl\{\tb^*(\z)(X)\bigr\}
=2\rho (\z)Z^{\kappa}\{X\}\,,\label{maint}\\
&Z^{\kappa}\bigl\{\bb^*(\z)(X)\bigr\}
=\frac 1{2\pi i}\oint\limits _{\Gamma}
\omega (\z,\xi)
Z^{\kappa}\bigl\{\cb(\xi)(X)\bigr\}
\frac{d\xi^2}{\xi^2}\,,\label{mainb}\\
&Z^{\kappa}\bigl\{\cb^*(\z)(X)\bigr\}
=-\frac 1 {2\pi i}\oint\limits_{\Gamma}
\omega (\xi,\z)
Z^{\kappa}\bigl\{\bb(\xi)(X)\bigr\}
\frac{d\xi^2}{\xi^2}\,,
\label{mainc}
\end{align}
where $\Gamma$ goes around $\xi ^2=1$. 
In particular,
\begin{align}
&\rho (\z)=\frac{1}{2}
Z^{\kappa}\bigl\{\tb^*(\z)(q^{2\al S{(0)}})\bigr\},\nn
\\
& 
\omega (\z,\xi)
=Z^\kappa\left\{\bb^*(\z)\cb^*(\xi)(q^{2\al S(0)})\right\}\,.\nn
\end{align}
They
are given in terms of the eigenvalues of the transfer matrices
and the $Q$ operators, as well as other characteristics 
in the Matsubara direction. Their explicit formulas will be 
given in \eqref{defrho} and \eqref{def-omega2}. 
In Appendix C we shall explain that
$\omega (\z,\xi)$ is a quantum deformation of the canonical 
normalised second kind differential
on a hyperelliptic Riemann surface. 

From the equations (\ref{maint}), (\ref{mainb}), (\ref{mainc})
one immediately derives 
\begin{align}
&Z^{\kappa}\Bigl\{\tb^*(\z _1^0)\cdots
\tb^*(\z _k^0)\bb^*(\z _1^+)\cdots
\bb^*(\z _l^+ )
\cb^*(\z _l^- )\cdots\cb^*(\z _1^-)
\bigl(q^{2\al S(0)}    \bigr)  \Bigr\}\label{detform}\\
&\qquad\qquad
=\prod\limits _{p=1}^k 2\rho (\z ^0_p)\times  
\det\left(\omega(\z _i^+,\z _j^-)
\right)_{i,j=1,\cdots ,l}\,.\nn
\end{align}
Taking the Taylor coefficients in $(\z^{\epsilon}_{i})^2-1$ 
in both sides, one obtains the value of 
$Z^{\kappa}$ on an arbitrary element of the basis (\ref{basis}). 
This is the main result of the paper. 
\medskip

The text is organised as follows. 

In Section 2 we give the precise 
definition of the linear functional 
$Z^{\kappa}$ on the space $\mathcal{W}_{\al,0}$. 
We explain that on any particular $X\in \mathcal{W}_{\al,0}$ 
this functional reduces to a finite expression.

In Section 3 we prove (\ref{maint}). 
A significant part of this section is devoted 
to the reduction of $Z^{\kappa}\bigl\{\tb^*(\z)(X)\bigr\}$ to finite intervals. This is a point which is used in Section 6.

In Section 4 we explain some simple 
facts about transfer matrices and $Q$ operators in the Matsubara
direction. It should be considered as 
preparation for the following sections.

In Section 5 we introduce $q$-deformed 
Abelian integrals which are constructed 
via eigenvalues of $Q$ operators in the Matsubara
direction. We introduce $q$-deformed exact forms and present the 
$q$-deformed Riemann bilinear relations.

In Section 6 we consider 
$Z^{\kappa}\bigl\{\bb^*(\z)(X)\bigr\}$. 
We formulate two lemmas 
which are proved in Appendix A and Appendix B. 
Informally, these lemmas say that 
$Z^{\kappa}\bigl\{\bb^*(\z)(X)\bigr\}$ 
is a $q$-deformation of a 
normalised second kind Abelian differential 
in $\z$, which has a prescribed singularity specified
by the quasi-local operator $X$. 
In the classical limit, such a differential 
can be expressed using the canonical normalised 
second kind differential. 
Formula (\ref{mainb}) is an analogue 
in the quantum case, the function $\omega (\z,\xi)$ playing the role  of the 
canonical differential.

In Section 7 we define $\omega (\z,\xi)$.  
Using the results in Section 5, 
we prove that it satisfies 
all the necessary requirements. 

Finally, in Section 8 we prove
the main Theorem 
which states that (\ref{mainb}), (\ref{mainc}) hold.

As mentioned above, Appendices A, B are devoted to the 
proof of the technical Lemmas in Section 6. 
In Appendix C we consider the classical limit 
of the $q$-deformed Abelian integrals and differentials. 
Then we explain that the classical
limit of $\omega (\z,\xi)$ is indeed related to 
the canonical normalised second kind
differential. 
Some general information about 
differentials on Riemann surfaces is provided. 
Readers who are not familiar with 
Riemann surfaces are recommended to read 
Section 5 and Appendix C together. 
In  Appendix D we show equivalence of several non-degeneracy
conditions accepted in the text.

\section{Definition of the linear functional $Z^{\kappa}$}
\label{definitionZ}

Consider a two dimensional finite lattice 
composed of two one-dimensional chains: 
the space chain, and the imaginary time or the 
Matsubara chain. 
The space chain has $2l$ sites which are
labelled by the letters $j=-l+1,\cdots,l$. 
With every site the 
Pauli matrices $\sigma ^a_j$ are associated. 
The Matsubara chain has $\mathbf{n}$ sites 
labelled by boldface letters $\mathbf{m}=\mathbf{1,\cdots , n}$.  
With every site we associate a half-integral 
spin $s_{\mathbf{m}}$ and a parameter $\tau _{\mathbf{m}}$, in other words 
a $(2s_{\mathbf{m}}+1)$-dimensional evaluation representation of 
the quantum group $U_q(\slth)$.
We assume that $\sum_{\mathbf{m}=1}^{\mathbf{n}} s_{\mathbf{m}}$ 
is an integer. 

We define the monodromy matrix
$$
T_{j,\mathbf{M}}(\z)
=L_{j,\mathbf{n}}(\z/\tau _{\mathbf{n}})
L_{j,\mathbf{n-1}}(\z/\tau _{\mathbf{n-1}})
\cdots 
L_{j,\mathbf{1}}(\z/\tau _{\mathbf{1}})
\,.
$$
The $L$ operator $L_{j,\mathbf{m}}(\z/\tau _{\mathbf{m}})$
is obtained from the universal one
$$ 
L_j(\z)=q^{\frac1 2}\begin{pmatrix}
\z ^2q^{\frac{H+1}2}-q^{-\frac{H+1}2}&(q-q^{-1})\z Fq^{\frac{H-1}2}\\
(q-q^{-1})\z q^{-\frac{H-1}2}E &
\z ^2q^{-\frac{H-1}2}-q^{\frac{H-1}2}
\end{pmatrix}_j\,, 
$$
by letting 
$E$, $F$, $H$ act on the $(2s_{\mathbf{m}}+1)$-dimensional representation
of $U_q(\slt)$. 
We shall consider a twisted transfer matrix
\be
&&T_{\mathbf{M}}(\z ,\kappa)
=\Tr _j\bigl(T_{j,\mathbf{M}}(\z ,\kappa)
\bigr),  
\\
&&
T_{j,\mathbf{M}}(\z ,\kappa)=
T_{j,\mathbf{M}}(\z)q^{\kappa\sigma ^3_j}, 
\en
and use the letter $T(\z,\kappa)$ to denote its eigenvalues.

Now we are ready to introduce the main object of our study. On the 
space  $\mathcal{W}_{\al,0}$  consider the linear functional
\begin{align}
Z^{\kappa}\Bigl\{ q^{2\al S(0)}\mathcal{O}  \Bigr\}=\lim_{l\to\infty}
\frac {\Tr _{\mathbf{M}}\Tr _{[-l+1,l]}\(
T_{[-l+1,l],\mathbf{M}}
\ q^{2(\kappa S_{[-l+1,l]}+\al S_{[-l+1,0]})}\mathcal{O}
\)}
{\Tr _{\mathbf{M}}\Tr _{[-l+1,l]}\(T_{[-l+1,l],\mathbf{M}}
\ q^{2(\kappa S_{[-l+1,l]}+\al S_{[-l+1,0]})}\)}\,.
\label{partition}
\end{align}
Here and for later use, we set
\bea
T_{[k,m],\mathbf{M}}=T_{k,\mathbf{M}}\cdots T_{m,\mathbf{M}}\,,
\quad
T_{j,\mathbf{M}}=T_{j,\mathbf{M}}(1)\,.
\label{monodromy-matrix}
\ena
In terms of the equivalent six-vertex model,  
the functional \eqref{partition} is given
by the following partition function on the infinite cylinder:
\vskip .5cm
\hskip .5cm\includegraphics[height=7cm]{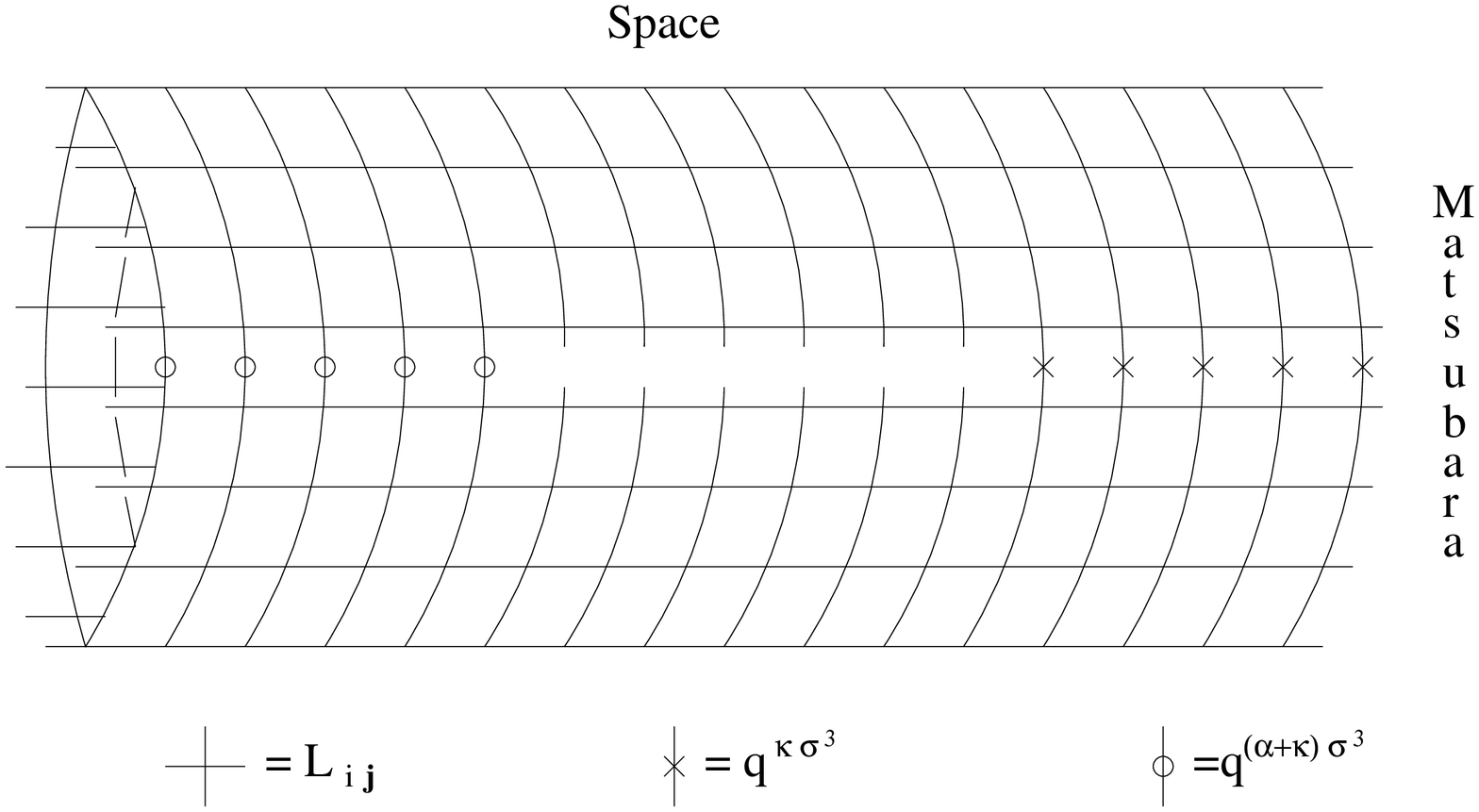}
\hfill{\it fig. 1}
\vskip .5cm
In this picture the broken links represent the operator $\mathcal{O}$: the
arrows on them are fixed.

Suppose that the transfer matrix 
$T_{\mathbf{M}}(1,\kappa)$ 
has a unique eigenvector $|\kappa\rangle$
such that the corresponding eigenvalue $T(1 ,\kappa)$ 
has the maximal absolute value. Similarly let
$\langle\kappa+\al |$ is be eigencovector 
of $T_{\mathbf{M}}(1, \kappa +\al)$ with
the eigenvalue $T(1, \kappa +\al)$ possessing the same property.
Let us remark that for the XXX model the spectrum in spin zero sector
is known to be simple even in the homogeneous case \cite{Tarasov}
even in the homogeneous case.
Suppose also that 
\begin{align}
\langle \kappa+\al|\kappa\rangle\ne 0\,.
\label{general}
\end{align}
It is clear that in this situation
\eqref{partition} reduces to the linear functional 
\begin{align}
Z^{\kappa}\Bigl\{ q^{2\al S(0)}\mathcal{O}  \Bigr\}=\lim_{l\to\infty}
\frac {\langle \kappa+\al|\Tr _{[-l+1,l]}\(T_{[-l+1,l],\mathbf{M}}\ q^{2(\kappa S_{[-l+1,l]}+\al S_{[-l+1,0]})}\mathcal{O}\)|\kappa\rangle}
{\langle \kappa+\al|\Tr _{[-l+1,l]}\(T_{[-l+1,l],\mathbf{M}}\ q^{2(\kappa S_{[-l+1,l]}+\al S_{[-l+1,0]})}\)|\kappa\rangle}\,.
\label{defZ1}
\end{align}
This is the object which we shall calculate. 
For any given quasi-local operator we can proceed further.
Indeed, if the support of 
$q^{2\al S(0)}\mathcal{O}=q^{2\al S(k-1)}X_{[k,m]} $
is contained in the interval $[k,m]$ of the space chain, then
\begin{align}
Z^{\kappa}\Bigl\{ q^{2\al S(k-1)}X_{[k,m]}  \Bigr\}=\rho(1)^{k-1}
\frac{\langle \kappa+\al|\Tr _{[k,m]}\(T_{[k,m],\mathbf{M}}
q^{2\kappa S_{[k,m]}}X_{[k,m]}\)|\kappa\rangle}{T(1,\kappa)^{m-k+1}
\ \langle \kappa+\al|\kappa\rangle}
\label{defZ2}\,,
\end{align}
where 
\begin{align}
\rho (\z )=\frac {T(\z,\al+\kappa)}{T(\z,\kappa)}
\label{defrho}\,. 
\end{align}
The function \eqref{defrho} will play an important  
role for us; we shall see in the next Section 
that this is the same function as in (\ref{maint}).
The last formula \eqref{defZ2}
shows, as it has been said,  that
the limit $l\to\infty$ is superfluous.  
It is put in the formula (\ref{defZ1}) 
just for the sake of treating all quasi-local operators simultaneously. 

It may look surprising that the thermodynamic limit 
in this approach is so simple. 
Usually, it requires  a complicated analysis of 
Bethe equations. 
Certainly, the complexity of the problem cannot disappear, 
and it is hidden in the 
limiting process $\mathbf{n}\to \infty$ to arrive at 
(\ref{TEV}). 
But the idea used in \cite{KMT}, 
and developed further in \cite{GKS},  
is that one can proceed 
rather far before taking this limit. 
This is especially true in the present work. 
The complexities of the thermodynamic limit of 
$Z^{\kappa}(X)$ are confined to only two functions,  
for which one can 
take the limit $\mathbf{n}\to \infty$ rather easily.

Let us emphasise one point which may be a source of confusion.
We started with (\ref{partition}), reduced it to (\ref{defZ1}) and
further to (\ref{defZ2}). The expression in the right hand side of
(\ref{defZ1}) is perfectly well defined for 
any pair of eigenvectors
of $T_\mathbf{M}(\z,\kappa)$ and 
$T_\mathbf{M}(\z,\kappa+\al)$ satisfying 
the condition (\ref{general}). 
For the computation of (\ref{defZ1})
we shall use only quite general facts 
concerning Bethe vectors, so they are valid in general. 
Still, the subject of our study is (\ref{partition}), 
and it reduces to (\ref{defZ1})
only for the eigenvectors corresponding to the  maximal eigenvalues. 

\section{Computation of $Z^{\kappa}\bigl\{\tb ^*(\z)\bigl( X\bigr)  \bigr\}$.}\label{compt}

According to (\ref{blocks}) we are actually interested only in the following
block of $\tb^*(\z)$:
$$\left.\tb^*(\z,\al)=\tb^*(\z)\right|_{\mathcal{W}_{\al,0}\to\mathcal{W}_{\al,0}}\,.$$
Let us recall the definition of the operator $\tb ^*(\z,\al)$  
given in the paper \cite{HGSII}. 
We start with a finite interval and  
an operator $X_{[k,m]}$. With this notation 
we imply that $X_{[k,m]}$ acts as $I$ outside $[k,m]$.
Define for $l>m$ $$
\tb ^*_{[k,l]}(\z,\al)(X_{[k,m]})=\Tr_a\(\mathbb{T}_{a,[k,l]}(\z,\al)(X_{[k,m]})\)\,,$$
where
\begin{align}
&\mathbb{T}_{a,[k,l]}(\z,\al)(X_{[k,m]})=T_{a,[k,l]}(\z)q^{\al\sigma ^3_a}X_{[k,m]}T_{a,[k,l]}(\z)^{-1}\,,\nn\\
&T_{a,[k,l]}(\z)=R_{a,l}(\z)\cdots R_{a,k}(\z)\,,\nn
\end{align}
$R_{a,j}(\z)$ is the standard 4 by 4 $R$-matrix
(see e.g. (2.4), \cite{HGSII}). Define further
$$\widetilde\R^\vee_{i,j}(\z^2) =\z^{\bS_i}\R_{i,j}(\z)\mathbb{P}_{i,j}\z^{-\bS_j}=1+(\z^2-1)\rb_{i,j}(\z^2)
\,,$$
where $\mathbb{P}_{i,j}(\cdot)=P_{i,j}(\cdot )P_{i,j}$ 
and $P_{i,j}$ is the permutation operator. 
Since $ \widetilde\R^\vee_{i,j}(1)=1$, 
$\rb_{i,j}(\z^2)$ is regular at $\z ^2=1$. 
Then \cite{HGSII}
\be
&&\tb^*_{[k,l]}(\z,\al)(X_{[k,m]})
=2\sum_{j=m}^{l-1}(\z^2-1)^{j-m}
\rb_{j+1,j}(\z^2)\cdots\rb_{m+2,m+1}(\z^2)
\tilde \R^\vee(\z^2)(Y_{[k,m+1]})\\
&&\qquad+(\z^2-1)^{l-m}\Tr_a\left\{\rb_{a,l}(\z^2)\rb_{l,l-1}(\z^2)\cdots
\rb_{m+2,m+1}(\z^2)\tilde\R^\vee(\z^2)(Y_{[k,m+1]})\right\}\,,
\en
where $Y_{[k,m+1]}=q^{\al\s^3_k}\taub(X_{[k,m]})$, 
$\taub$ is the shift by one site of the chain to the right, and 
\be
\tilde \R^\vee(\z^2)(Y_{[k,m+1]})=\widetilde\R^\vee_{m+1,m}(\z^2)
\cdots\widetilde\R^\vee_{k+1,k}(\z^2)(Y_{[k,m+1]}).
\en
Hence the limit $l\to\infty$ is well defined as a 
power series in $\z ^2-1$:
\begin{align}
\tb^*(\z,\al)&\bigl(q^{2\al S(k-1)}X_{[k,m]}\bigr)
=\lim_{l\rightarrow\infty}
q^{2\al S(k-1)}\tb^*_{[k,l]}(\z,\al)(X_{[k,m]})\nn\\
&=2q^{2\al S(k-1)}\sum_{j=m}^\infty(\z^2-1)^{j-m}\rb_{j+1,j}(\z^2)
\cdots\rb_{m+2,m+1}(\z^2)\tilde \R^\vee(\z^2)(Y_{[k,m+1]})\,.\nn
\end{align}

We repeated these definitions because we want to make clear
the following point. Take a 2 by 2 matrix $K$ such that $\Tr(K)\ne 0$ and
consider the following object:
$$\tb ^*_{[k,l]}(\z,\al,K)(X_{[k,m]})=\frac 2 {\Tr(K)}
\Tr_a\(K_a\mathbb{T}_{a,[k,l]}(\z,\al)(X_{[k,m]})\)\,,$$
Then it is easy to conclude from the above definition that
\begin{align}
\tb ^*_{[k,l]}(\z,\al,K)(X_{[k,m]})=\tb ^*_{[k,l]}(\z,\al)(X_{[k,m]})
\quad \text{mod}\ (\z^2-1)^{l-m}\,.\label{K}
\end{align}

\begin{lem}\label{thm:t*rho}
We have
\begin{align}
Z^{\kappa}\Bigl\{\tb ^*(\z)\bigl(q^{2\al S(0)}\mathcal{O}\bigr)\Bigr\}=
2\rho (\z)Z^{\kappa}\Bigl\{q^{2\al S(0)}\mathcal{O}\Bigr\}\,.
\label{Zt*}
\end{align}
\end{lem}
\begin{proof}
Without loss of generality, 
let $\mathcal{O}=X_{[1,m]}$ be localised on the interval $[1,m]$.  
\begin{align}
&Z^{\kappa}\Bigl\{
\tb ^*(\z,\al)\bigl(X_{[1,m]}q^{2\al S(0)}\bigr)
\Bigr\}\nn\\&=
\lim _{l\to \infty}
\frac{\langle\kappa+\al|\Tr _{[1,l],a}\Bigl(T_{[1,l],\mathbf{M}}
q^{2\kappa S_{[1,l]}}
\mathbb{T}_{a,[1,l]}(\z,\al)(X_{[1,m]})\Bigr)|\kappa\rangle}{T(1,\kappa)^{l}
\langle\kappa+\al|\kappa\rangle}\,.\nn
\end{align}
From the considerations  above we obtain
\begin{align}
&\langle\kappa+\al|\Tr _{[1,l],a}\Bigl(T_{[1,l],\mathbf{M}}q^{2\kappa S_{[1,l]}}
\mathbb{T}_{a,[1,l]}(\z)(X_{[1,m]})\Bigr)|\kappa\rangle
\nn\\
&=\frac 2 {T(\z,\kappa)}
\langle\kappa+\al|\Tr _{[1,l],a}\Bigl(T_{[1,l],\mathbf{M}}q^{2\kappa S_{[1,l]}}
T_{a,\mathbf{M}}(\z)q^{\kappa\sigma^3_a}\mathbb{T}_{a,[1,l]}(\z,\al)(X_{[1,m]})\Bigr)|\kappa\rangle\nn\\
&\qquad\qquad\qquad\qquad\qquad\qquad\qquad\qquad\qquad\qquad\qquad\qquad\qquad \text{mod}\ (\z^2-1)^{l-m}\,.
\nn
\end{align}
The idea here is exactly as in (\ref{K}). 
The monodromy matrix $T_{a,\mathbf{M}}(\z)q^{\kappa\sigma ^3_a}$ 
plays the role of $K_a$. The fact
that it carries the additional structure as operator
in the Matsubara space is not important.  
What is important is
that the state $|\kappa\rangle$ 
is an eigenstate of 
$\Tr _a\bigl(T_{a,\mathbf{M}}(\z)q^{\kappa\sigma ^3_a}\bigr)$ 
with eigenvalue $T(\z,\kappa)$. 
Now we can proceed using the Yang-Baxter equation and the 
cyclicity of trace:
\begin{align}
&\frac 2 {T(\z,\kappa)}
\langle\kappa+\al|\Tr _{[1,l],a}\Bigl(T_{[1,l],\mathbf{M}}\ q^{2\kappa S_{[1,l]}}
T_{a,M}(\z)q^{\kappa\sigma^3_a}\mathbb{T}_{a,[1,l]}(\z,\al)(X_{[1,m]})\Bigr)|\kappa\rangle\nn\\
&=\frac 2 {T(\z,\kappa)}
\langle\kappa+\al|\Tr _{[1,l],a}\Bigl(\mathbb{T}_{a,[1,l]}(\z)\Bigl(T_{a,\mathbf{M}}(\z )q^{(\kappa+\al)\sigma ^3_a}T_{[1,l],\mathbf{M}}\ q^{2\kappa S_{[1,l]}}
X_{[1,m]}\Bigr)\Bigr)|\kappa\rangle\nn\\
&=\frac 2 {T(\z,\kappa)}
\langle\kappa+\al|\Tr _{[1,l],a}\Bigl(T_{a,\mathbf{M}}(\z )q^{(\kappa+\al)\sigma ^3_a}T_{[1,l],\mathbf{M}}\ q^{2\kappa S_{[1,l]}}
X_{[1,m]}\Bigr)|\kappa\rangle\nn\\
&=2\rho(\z)
\langle\kappa+\al|\Tr _{[1,l]}\Bigl(T_{[1,l],\mathbf{M}}\ q^{2\kappa S_{[1,l]}}
X_{[1,m]}\Bigr)|\kappa\rangle\nn\,,
\end{align}
which proves the assertion. 
\end{proof}

Some comments on (\ref{Zt*}) have to be made. It has been 
said that $\taub=\tb^*_1/2$ is the shift by one site of the chain 
to the right. According to \cite{HGSII} the rest of $\tb^*_p$ is constructed
from the adjoint action of local integrals of motion. 
Then, looking at (\ref{partition}) 
one may wonder where $\rho(\z)$ comes from. 
Naively, it should not be in the right hand side 
because
$\taub$ and adjoints of the integrals of motion commute with 
$\Tr _{\mathbf{M}}(T_{Q,\mathbf{M}})$ and hence they should not contribute to (\ref{partition}) due to the 
cyclicity of trace. 
However, this is not correct 
in the presence of disorder field $q^{2\al S(0)}$. 
Let us explain this point in the simplest case  
$\taub=\tb^*_1/2$. Consider the definition
(\ref{partition}). For finite $l$ in (\ref{partition}), 
we define the cyclic shift by one site $\taub^{\text{periodic}}$,  
which acts in particular as 
$$
\taub^{\text{periodic}}
\Bigl(q^{2\al S_{[-l+1,0]}}\Bigr)= q^{2\al S_{[-l+2,1]}}\,.
$$
On the other hand, it is easy to see from the definition 
that  our operator $\taub$ acts as
$$
\taub\Bigl(q^{2\al S_{[-l+1,0]}}\Bigr)=  q^{2\al S_{[-l+1,1]}}\,.
$$
This difference accounts for the appearance of $\rho(1)$
in the functional (\ref{partition}). 
A similar thing happens with the 
adjoint action of the local integral of motion $\mathbb{I}_p(\cdot)=
[I_p,\cdot ]$. 
The operator $\mathbb{I}_{p}^{ \text{periodic}}$ feels 
the two inhomogeneities 
of $q^{2\al S_{[-l+1,0]}}$: between sites $0$ and $1$ 
and between sites $-l+1$ and $l$, while
the operators entering the definition of $\tb^*(\z)$ 
feel only the first one. 
This is the reason why $\rho(\z)$ appears. 
There are two cases when $\rho(\z)=1$. 
The first one is trivial: $\al=0$. The second one
is the case of VEVs (\ref{exp}) which was 
considered in \cite{HGS}, \cite{HGSII}.

Before proceeding to $\bb ^*$ and $\cb ^*$ we have to give some explanation
about $q$-deformed Abelian integrals.

\section{Spectral properties in Matsubara direction.}\label{SPMD}

Consider the transfer matrix
$$T_{\mathbf{M}}(\z ,\la)=\Tr _a\bigl(T_{a,\mathbf{M}}(\z)q^{\la\sigma ^3_a}\bigr)\,.$$ 
Let us introduce the $Q$ operator 
\begin{align}
Q^{+}_{\mathbf{M}}(\z ,\la)=\z ^{\la-\mathbf{S}}\Tr _A\(T_{A,\mathbf{M}}(\z ,\la)\)
\nn\,,
\end{align}
where $\mathbf{S}$ is the total spin operator acting on the Matsubara chain, 
and
$$
T_{A,\mathbf{M}}(\z ,\la)=L_{A,\mathbf{n}}(\z/\tau _{\mathbf{n}})
\cdots  L_{A,\mathbf{1}}(\z/\tau _{\mathbf{1}})q^{2\la D_A}\,.
$$
Here the $L$ operators are associated with the $q$-oscillator algebra
with generators $\ab _A, \ab _A^* ,D_A$. 
For the notation and conventions, see \cite{HGSII}. 
If $s_{\mathbf{m}}=1/2$, then 
\begin{align}
L_{A,\mathbf{m}}(\z)=\begin{pmatrix}
1-\z ^2q^{2D_A+2} &-\z\ab _A\\
-\z\ab _A^* &1
\end{pmatrix}_{\mathbf{m}}
\begin{pmatrix}
q^{-D_A} &0\\ 0 &q^{D_A}
\end{pmatrix}_{\mathbf{m}}\,. 
\label{Lop}\end{align}
To obtain $L_{A,\mathbf{m}}(\z)$ 
for other spins, one applies the standard fusion procedure.
The $Q$ operator $Q^{-}_{\mathbf{M}}(\z,\la)$ is defined by
\begin{align}
Q^{-}_{\mathbf{M}}(\z,\la )=J\ Q^{+}_{\mathbf{M}}(\z ,-\la)\ J\,,\nn
\end{align}
where $J$ is the operator of spin reversal. 
 
These $Q$ operators satisfy the Baxter equation:
\begin{align}
&T_{\mathbf{M}}(\z,\la)Q^{\pm}_{\mathbf{M}}(\z ,\la)=d(\z)Q^{\pm}_{\mathbf{M}}(\z q ,\la)
+
a(\z)Q^{\pm}_{\mathbf{M}}(\z q^{-1}, \la)
\,.\label{Baxter}
\end{align}
These equations hold for the eigenvalues because $T_{\mathbf{M}}(\z,\la)$
commute with $Q^{\pm}_{\mathbf{M}}(\xi ,\la)$. 
 
The functions $a(\z)$, $d(\z)$  are defined by spins and inhomogeneities present in the 
Matsubara direction. 
\begin{align}
&a(\z)=\prod\limits _{\mathbf{m=1}}^{\mathbf{n}}
a_{s_{\mathbf{m}}}(\z/\tau_{\mathbf{m}}),\qquad
a_s(\z)=\z^2q^{2s+1}-1\,,\label{eq-ad}\\
& d(\z)=\prod\limits _{\mathbf{m=1}}^{\mathbf{n}}
d_{s_{\mathbf{m}}}(\z/\tau_{\mathbf{m}}),\qquad
d_s(\z)=\z^2q^{-2s+1}-1
\,.\nn
\end{align}
Let us cite one formula from \cite{BLZ}:
\begin{align}
Q^{+}_{\mathbf{M}}(\z ,\la)Q^{-}_{\mathbf{M}}(\z q,\la)
-Q^{-}_{\mathbf{M}}(\z ,\la)Q^{+}_{\mathbf{M}}(\z q,\la)
=\frac 1 {q^{\la  -\mathbf{S}}-q^{-\la +\mathbf{S}}}
W(\z),\
\label{wronskian}
\end{align}
where 
$$
W(\z)=\prod\limits _{\mathbf{m=1}}^{\mathbf{n}}
w_{s_{\mathbf{m}}}(\z /\tau_{\mathbf{m}} ),
\quad w_{s}(\z )=\prod\limits _{k=1}^{2s}(1-\z ^2q^{2k-2s+1})
\,.
$$

Suppose that $T_{\mathbf{M}}(\z,\la)$ has a unique eigenvector $|\la\rangle$ 
with eigenvalue $T(\z,\la)$ such that $T(1,\la)$ 
has  maximal absolute value.
We denote by 
$Q^{\pm}(\xi ,\la)$ the eigenvalues of 
$Q^{\pm}_{\mathbf{M}}(\xi ,\la)$ on $|\la \rangle$.
If the eigenvector $|\la \rangle$ has spin 
$d-\sum\limits _{\mathbf{m}=1}^{\mathbf{n}} s_{\mathbf{m}}$, 
it follows 
from the form of the $L$ operator (\ref{Lop})
that $\z^{-\la+\mathbf{S}}Q^{+}(\z ,\la)$
is a polynomial in $\z^2$ of degree $d$, 
while $\z ^{\la-\mathbf{S}}Q^{-}(\z ,\la)$ 
is of degree 
$2\sum\limits _{\mathbf{m}=1}^{\mathbf{n}} s_{\mathbf{m}}-d$.
Due to the quantum Wronskian relation (\ref{wronskian}),  
their leading and the lowest coefficients 
are both nonzero.

Let us discuss the symmetry under negating $\la$. We have
\begin{align}
T_{\mathbf{M}}(\z,-\la)=J\ T_{\mathbf{M}}(\z,\la)\ J\,,\label{hhh}
\end{align}
which implies that the spectra of $T_{\mathbf{M}}(\z,\la)$ and $T_{\mathbf{M}}(\z,-\la)$ coincide, 
and, in particular,
\begin{align}
T(\z,\la)=T(\z,-\la)\,.\label{symT}
\end{align}
Furthermore, the equation
$$
Q^-_{\mathbf{M}}(\z,\la)=J\ Q^+_{\mathbf{M}}(\z,-\la)\ J
$$
implies that
\begin{align}
Q^-(\z,\la)=Q^+(\z,-\la)\,.\label{symQ}
\end{align}
Due to \eqref{hhh} the vectors $|\la\rangle $ and $|-\la\rangle$ have opposite spins.

\section{Deformed Abelian integrals}\label{DAI}

Working with quantum integrable models, 
one should not neglect the important piece of intuition  provided 
by the method of Separation of Variables (SoV) discovered by
Sklyanin \cite{Skl}. It has been explained in \cite{SmiToda} that
the matrix elements of observables in the SoV method are
expressed in terms of deformed Abelian integrals. In the
case under consideration, which is related to the algebra
$U_q(\slth)$, these integrals are deformations of hyperelliptic ones.
Let us give their definition. 

Introduce the function $\varphi (\z)$ which
satisfies the equation
\begin{align}
a(\z q) \varphi (\z q)=d(\z)\varphi (\z)\,.
\label{eqphi}
\end{align}
This function is elementary, 
$$
\varphi (\z)
=\prod\limits_{\mathbf{m=1}}^{\mathbf{n}}
\varphi_{s_{\mathbf{m}}}(\z/\tau _{\mathbf{m}})\,,\quad
\varphi_s(\z)=\prod\limits _{k=0}^{2s}\frac 1
{\z ^2q^{-2s+2k+1}-1}\,.
$$
In addition to the contour $\Gamma$ which encircles $\z^2=1$, 
we consider $\mathbf{n+1}$ contours in the $\z ^2$ plane:   
$\Gamma _{\mathbf{0}}$ which goes around $0$, 
and $\Gamma _{\mathbf{m}}$ which encircles
the poles $\z^2= \tau ^2_{\mathbf{m}}q^{2s_{\mathbf{m}}-2k-1}$
($k=0,\cdots, 2s_{\mathbf{m}}$)
of $\varphi_{s_{\mathbf{m}}}(\z/\tau _{\mathbf{m}})$. 

In the following, we use the $q$-difference operator
$$
\Delta _{\z}f(\z)=f(\z q)-f(\z q^{-1})\,.
$$
It acts on the class of functions of the form 
$f(\z)=\z^{\la}P(\z^2)$, $P$ being a polynomial in $\z^2$
and $q^{2(n+\la)}\neq1$ for all integers $n$.    
Within this class the $q$-primitive $\Delta^{-1}_\z f(\z)$ 
is defined uniquely. 

There are two kinds of deformed Abelian integrals,  
\begin{align}
\int\limits _{\Gamma_{\mathbf{m}}}f^{\pm}(\z)Q^{\mp}(\z ,\kappa+\al)Q^{\pm}(\z,\kappa)
\varphi (\z)\frac{d\z ^2}{\z ^2}\,,\label{defAI}
\end{align}
where $\z^{\mp\al}f^{\pm}(\z)$ is a  polynomial in $\z^2$, in order that
the integrand is single-valued. 

We start our study of deformed Abelian integrals with the 
following technical lemma. 
\begin{lem}\label{1stAI}
Let $\z^{\mp\al}f^{\pm}(\z)$ be a polynomial in $\z^2$. Then,  
for $\mathbf{m}=\mathbf{0},\cdots,\mathbf{n}$, the following identities hold. 
\begin{align}
&\int\limits _{\Gamma_{\mathbf{m}}}
\Bigl\{T(\z,\kappa)\Delta_\z^{-1}f^{\pm}(\z q)
 -T(\z,\kappa+\al)\Delta_\z^{-1}f^{\pm}(\z)
\Bigr\}
Q^{\mp}(\z ,\kappa+\al)Q^{\pm}(\z,\kappa)
\varphi (\z)\frac{d\z ^2}{\z ^2}
\label{AIa}\\
&=\int\limits _{\Gamma_{\mathbf{m}}}
f^{\pm}(\z)a(\z)Q^{\mp}(\z,\kappa+\al)Q^{\pm}(\z q^{-1},\kappa)
\varphi (\z)\frac{d\z ^2}{\z ^2}\,,
\nn
\end{align}
\begin{align}
&\int\limits _{\Gamma_{\mathbf{m}}}
\Bigl\{T(\z,\kappa +\al)\Delta_\z^{-1}f^{\pm}(\z )
 -T(\z,\kappa)\Delta_\z^{-1}f^{\pm}(\z q^{-1})
\Bigr\}
Q^{\mp}(\z ,\kappa+\al)Q^{\pm}(\z,\kappa)
\varphi (\z)\frac{d\z ^2}{\z ^2}
\label{AId}\\
&=\int\limits _{\Gamma_{\mathbf{m}}}
f^{\pm}(\z)d(\z)Q^{\mp}(\z ,\kappa+\al)Q^{\pm}(\z q,\kappa)
\varphi (\z)\frac{d\z ^2}{\z ^2}\,.
\nn
\end{align}
\end{lem}
\begin{proof}
This can be verified directly
by applying the Baxter equation to $T(\z,\kappa)Q^{\pm}(\z,\kappa)$,
$T(\z,\kappa +\al)Q^{\mp}(\z ,\kappa+\al)$ 
and moving contours of integration.
\end{proof}

It is well known that, on a compact Riemann surface of genus $g$, 
the space of the first and the second kind differentials 
(meromorphic differentials without residues) 
has a finite dimension $2g$ when considered 
modulo exact forms. In \cite{SmiToda} it was explained
what are the $q$-deformed exact forms, for which 
the deformed Abelian integrals vanish. 
Since the proof was omitted in that paper we include it here.

\begin{lem}\label{exact}
Define a $q$-deformed exact form to be an expression 
\begin{align}
&E\bigl(f^{\pm}(\z)\bigr)\label{exact form}\\
&=T(\z,\kappa)\Delta _{\z}^{-1}\(f^{\pm}(\z)T(\z,\kappa)\)
+T(\z,\kappa+\al)\Delta _{\z}^{-1}\(f^{\pm}(\z)
T(\z ,\kappa+\al)\)\nn\\
&-T(\z ,\kappa)\Delta _{\z}^{-1}\(f^{\pm}(\z q)
T(\z q, \kappa+\al)\)
 -T(\z ,\kappa+\al)\Delta _{\z}^{-1}
\(f^{\pm}(\z q^{-1})T(\z q^{-1},\kappa)\)\nn\\
&+a(\z q)d(\z)f^{\pm}(\z q)-d(\z q^{-1})a(\z)
f^{\pm}(\z q ^{-1})\,, 
\nn
\end{align}
where $\z ^{\mp\al}f^{\pm}(\z)$ is a polynomial in $\z^2$. 
Then we have 
\begin{align}
\int\limits _{\Gamma_{\mathbf{m}}}
E\bigl(f^{\pm}(\z)\bigr)Q^{\mp}(\z ,\kappa+\al)Q^{\pm}(\z,\kappa)
\varphi (\z)\frac{d\z ^2}{\z ^2}=0\,.\nn
\end{align}
\end{lem}
\begin{proof}
Let us divide the integral into two pieces:
\begin{align}
&\int\limits _{\Gamma_{\mathbf{m}}}
E\bigl(f^{\pm}(\z)\bigr)
Q^{\mp}(\z ,\kappa+\al)Q^{\pm}(\z,\kappa)
\varphi (\z)\frac{d\z ^2}{\z ^2}
=I_1+I_2\,,
\end{align}
where $I_1$ first four terms from \eqref{exact form} and $I_2$
contains remaining two. 
Apply (\ref{AIa}) to the first and the fourth terms
in $I_1$
, and 
to the second and the third terms as well.  
Using the Baxter equation 
and moving contours using (\ref{eqphi}), one obtains
\begin{align}
I_1&=\int\limits _{\Gamma_{\mathbf{m}}}
f^{\pm}(\z q^{-1})T(\z q^{-1},\kappa)
Q^{\mp}(\z ,\kappa+\al)Q^{\pm}(\z q^{-1},\kappa)a(\z)
\varphi (\z)\frac{d\z ^2}{\z ^2}
\nn\\
&-\int\limits _{\Gamma_{\mathbf{m}}}
f^{\pm}(\z )T(\z ,\kappa +\al)
Q^{\mp}(\z ,\kappa+\al)Q^{\pm}(\z q^{-1},\kappa)a(\z)
\varphi (\z)\frac{d\z ^2}{\z ^2}\nn\,.
\end{align}
Now apply the Baxter equation:
\begin{align}
I_1=&\int\limits _{\Gamma_{\mathbf{m}}}
f^{\pm}(\z q^{-1})
Q^{\mp}(\z ,\kappa+\al)
\(Q^{\pm}(\z q^{-2},\kappa)
a(\z q^{-1})+
Q^{\pm}(\z ,\kappa)d(\z q^{-1})\)
a(\z)
\varphi (\z)\frac{d\z ^2}{\z ^2}
\nn\\
&-\int\limits _{\Gamma_{\mathbf{m}}}
f^{\pm}(\z )
\(a(\z)Q^{\mp}(\z q^{-1} ,\kappa+\al)+d(\z)Q^{\mp}(\z q,\kappa+\al)\)Q^{\pm}(\z q^{-1},\kappa)
a(\z)
\varphi (\z)\frac{d\z ^2}{\z ^2}
\nn
\end{align}
Moving contours we find
\begin{align}
&I_1=\int\limits _{\Gamma_{\mathbf{m}}}
\Bigl\{a(\z q)d(\z)f^{\pm}(\z q)-d(\z q^{-1})a(\z)f^{\pm}(\z q ^{-1})   \Bigr\}
Q^{\mp}(\z ,\kappa+\al)Q^{\pm}(\z,\kappa)
\varphi (\z)\frac{d\z ^2}{\z ^2}
\,,\nn
\end{align}
i.e. $I_1=-I_2$.
\end{proof}

A beautiful feature of deformed Abelian integrals
is that they allow for a deformation of   
the Riemann bilinear relations as well.  
In \cite{SmiToda} the latter 
are given in the full-fledged form. 
For our present purposes, it is sufficient to 
use a part of them given by the following Lemma.

\begin{lem}
Consider the following function in two variables
\begin{align}
&r(\z,\xi)=r^+(\z,\xi)-r^-(\xi,\z)\,,\nn
\end{align}
where
$$r^+(\z,\xi)=r^+(\z,\xi|\kappa,\al),\quad 
r^-(\xi,\z)=r^+(\xi,\z|-\kappa,-\al),$$
and
\begin{align}
&r^+(\z,\xi|\kappa,\al)=
T(\z ,\kappa)\Delta _{\z}^{-1}\(\psi(\z/\xi,\al)(T(\z ,\kappa)-
T(\xi ,\kappa))\)\label{r+}\\&
+T(\z ,\kappa+\al)\Delta _{\z}^{-1}\(\psi(\z/\xi,\al)(T(\z ,\kappa+\al )
-T(\xi ,\kappa+\al ))\)\nn\\
&-T(\z ,\kappa)\Delta _{\z}^{-1}\(\psi(q\z/\xi,\al)(T(\z q
,\kappa+\al)-
T(\xi ,\kappa+\al))\)\nn\\ &
-T(\z ,\kappa+\al)\Delta _{\z}^{-1}\(\psi(q^{-1}\z/\xi,\al)(T(\z q^{-1} ,\kappa)
  -T(\xi,\kappa))\)\nn\\&+\bigl(a(\z q)-a(\xi)\bigr)d(\z)\psi(q\z/\xi,\al)
-\bigl(d(\z q^{-1}\bigr)
  -d(\xi))a(\z)\psi(q^{-1}\z/\xi,\al)\,.\nn
\end{align}
Then
\begin{align}
\int
\limits _{\Gamma _{\mathbf{i}}}\int\limits _{\Gamma _{\mathbf{j}}}
r(\z,\xi)
Q^{-}(\z,\kappa+\al)Q^{+}(\z,\kappa)
Q^{+}(\xi,\kappa+\al)Q^{-}(\xi,\kappa)
\varphi (\z)\varphi (\xi)\frac{d\z ^2}{\z ^2}\frac{d\xi ^2}{\xi ^2}=0\,.
\label{riemann}
\end{align}
\end{lem}
\begin{proof}
The proof is similar to that of the previous lemma.  
We apply Lemma \ref{1stAI}, invoke the
Baxter equation and move the contours. 
When the Baxter equation 
is applied 
to expressions like
$\psi(\z/\xi,\al)(T(\z ,\kappa)-T(\xi ,\kappa))$ separately with respect to $\z$ and $\xi $,
a singularity may appear from $\psi(\z/\xi,\al)$.  
In general, by moving the contours such a singularity produces 
intersection numbers 
as in the genuine Riemann bilinear relations (see \cite{SmiAI}). 
In the present case this does not happen
 because the contours do not have nontrivial intersections.
\end{proof}

Clearly $\xi^{\al}r^+(\z,\xi)$ 
is a polynomial in $\xi^2$ and 
$\z^{-\al}r^-(\xi,\z)$ is a polynomial 
in $\z^2$, both of degree $\mathbf{n}$. 
This allows us to define the polynomials $p^{\pm}_{\mathbf{m}}$ by 
\begin{align}
r^+(\z,\xi)
=\sum\limits_{\mathbf{m}=0}^{\mathbf{n}}
\z^{\al}p^+_ {\mathbf{m}}(\z ^2)
\xi ^{-\al +2\mathbf{m}}\,,
\qquad r^-(\xi,\z)=\sum\limits _{\mathbf{m}=0}^{\mathbf{n}}
\xi ^{-\al}p^-_ {\mathbf{m}}(\xi ^2)
\z ^{\al +2\mathbf{m}}\,. 
\nn
\end{align}
Introduce the $(\mathbf{n+1})\times (\mathbf{n+1})$ matrices
\begin{align}
&\mathcal{A}^\pm_{\mathbf{i},\mathbf{j}}
=\int\limits _{\Gamma _{\mathbf{i}}}
\z^{\pm \al+2\mathbf{j}}
Q^{\mp}(\z,\kappa+\al)
Q^{\pm}(\z,\kappa)
\varphi (\z)\frac{d\z ^2}{\z ^2}\,,
\label{AB1}\\
&\mathcal{B}^\pm_{\mathbf{i},\mathbf{j}}
=\int\limits_{\Gamma _{\mathbf{i}}}
\z ^{\pm \al}p^\pm_{\mathbf{j}}(\z^2)
Q^\mp(\z ,\kappa+\al)Q^\pm(\z, \kappa)
\varphi (\z)\frac{d\z ^2}{\z ^2}\,.
\label{AB2}
\end{align}
Then (\ref{riemann}) reads as 
\begin{align}
\mathcal{B}^+(\mathcal{A}^-)^t=\mathcal{A}^+(\mathcal{B}^-)^t\,.
\label{symplectic}
\end{align}
We explain in Appendix C that, in the classical limit $q\to 1$
(and for $\al=0$),    
$\mathcal{A}^\pm$, $\mathcal{B}^\pm$ reduce 
to the matrices of $a$-periods of differentials 
of the first and the second kind, respectively. 
The relation (\ref{symplectic}) becomes 
one quarter of the classical Riemann bilinear relations 
which state that the full matrix of $a$- and $b$-periods
is an element of the symplectic group. 

Before closing this section let us make a comment. 
Suppose $\z^{\mp \al}f(\z)$ 
is a rational function. We assume that the poles of this function
do not overlap those of $\varphi(\z)$ and $\z^2=0$.  
In this case, the $q$-primitive $\Delta_\z^{-1}f(\z)$ 
is not uniquely defined, and in general    
develops infinitely many poles $q^{2n}w$ ($n\in\Z$, $w$ are the poles of $\z^{\mp \al}f(\z)$ ). 
Nevertheless Lemma \ref{1stAI} remains true. 
Actually it tells that the deformed
Abelian integrals 
in the left hand side of \eqref{AIa}, \eqref{AId} do not depend on 
a particular choice of the $q$-primitive. 
For the same reason, deformed Abelian integrals of the 
$q$-exact form in Lemma \ref{exact} have
unambiguous meaning.
Later on we shall deal with examples of such 
$q$-primitives of the form
$\Delta^{-1}_\z(\psi(\z/\xi,\al)P(\z^2))$ or
$\Delta^{-1}_\z(\psi(\xi/\z,\al)P(\z^2))$.


\section{Properties of $Z^{\kappa}\bigl\{\bb ^*(\z)\bigl( X\bigr)  \bigr\}$ and
$Z^{\kappa}\bigl\{\cb ^*(\z)
X\bigr)  \bigr\}$.}\label{propbc}

Our strategy is to compute 
$Z^{\kappa}\bigl\{\bb ^*(\z)\bigl( X\bigr)  \bigr\}$
and
$Z^{\kappa}\bigl\{\cb ^*(\z)
X\bigr)  \bigr\}$
inductively, reducing them to 
similar quantities involving the 
annihilation operators $\bb(\z),\cb(\z)$. 
It has been said in Introduction 
that $Z^{\kappa}\bigl\{\bb ^*(\z)\bigl( X\bigr)  \bigr\}$
is non-trivial only when $X\in \mathcal{W}_{\al+1,-1}$, and 
$Z^{\kappa}\bigl\{\cb ^*(\z)\bigl( X\bigr)  \bigr\}$
is non-trivial only when $X\in \mathcal{W}_{\al-1,1}$. 
We denote these blocks by 
\begin{align}
&\left.\bb ^*(\z,\al)=\bb ^*(\z)\right|_{\mathcal{W}_{\al+1,-1}\to \mathcal{W}_{\al,0}}\,,\quad
\left.\cb (\z,\al)=\cb (\z)\right|_{\mathcal{W}_{\al+1,-1}\to \mathcal{W}_{\al,0}}\,,\nn\\
&\left.\cb ^*(\z,\al)=\cb ^*(\z)\right|_{\mathcal{W}_{\al-1,1}\to \mathcal{W}_{\al,0}}\,,\quad
\ \ \left.\bb (\z,\al)=\bb (\z)\right|_{\mathcal{W}_{\al-1,1}\to \mathcal{W}_{\al,0}}\,.\nn
\end{align}

Hence the non-trivial part of our main equations 
(\ref{mainb}), (\ref{mainc}) 
take the form
\bea
&&Z^{\kappa}\bigl\{\bb ^*(\z,\al)\bigl(X\bigr)\bigr\} 
=\frac 1 {2\pi i}\oint\limits_{\Gamma}\omega(\z,\xi)
Z^{\kappa}\bigl\{\cb(\xi,\al)\bigl(X\bigr)\bigr\}
\frac{d\xi^2}{\xi^2},\quad\ \ \ X\in \mathcal{W}_{\al+1,-1}\,,
\label{b*to c1}\\
&&Z^{\kappa}\bigl\{\cb ^*(\z,\al)\bigl(X\bigr)\bigr\} 
=-\frac 1 {2\pi i}\oint\limits_{\Gamma}\omega(\z,\xi)
Z^{\kappa}\bigl\{\bb(\xi,\al)\bigl(X\bigr)\bigr\}
\frac{d\xi^2}{\xi^2}, \quad X\in \mathcal{W}_{\al-1,1}\,.
\label{b*to c2}
\ena
Our task is to establish the existence of $\omega(\z,\xi)$ and to determine it explicitly. 
In view of the spin reversal symmetry which relates 
$(\bb^*,\cb)$ with $(\cb^*,\bb)$, we shall concentrate on the first pair. 

Apart from an overall power of $\z$, 
$\bb^*(\z,\al)$ is defined a priori 
as a formal power series in $\z^2-1$. 
Nevertheless when acting on each operator 
it reduces to a rational function,  
due to the same mechanism as explained for $\tb^*$. 
Namely, 
\begin{align}
&\bb^*(\z,\alpha)(q^{2(\alpha+1)S(0)}X_{[1,m]})=q^{2\alpha S(0)}
\lim_{l\rightarrow\infty}\Tr_c\{\mathbb T_{c,[m+1,l]}(\z)\gb_{c,[1,m]}(\z,\alpha)(X_{[1,m]})\}.\label{redbstar}
\end{align}
We recall the definition of the operator
$\gb_{c,[1,m]}(\z,\al)$ in Appendix \ref{app:bstar1}. 
The formula  \eqref{redbstar} together
with the
requirement of translational invariance can be considered as a definition of $\bb^*(\z,\alpha)$,
but self-consistency
of this definition requires that $\gb_{c,[1,m]}(\z,\al)$ satisfies
certain reduction relations which were proved in \cite{HGSII}.

Using (\ref{redbstar}) we find by the same method as in
Lemma \ref{thm:t*rho}:
\bea
&&T(\z,\kappa)Z^\kappa\left(\bb^*(\z,\al)(q^{2(\al+1) S(0)}X_{[1,m]})\right)\label{redZb*}
\\
&&\quad =
\frac{\Tr_{[1,m], c}
\left(\langle{\kappa+\al}|T_{[1,m],\mathbf{M}}(1,\kappa)T_{c,\mathbf{M}}(\z,\kappa)
2\gb _{c,[1,m]}(\z,\al)(X_{[1,m]})
|{\kappa}\rangle
\right)}
{T(1,\kappa)^{m}\langle\kappa+\al|\kappa\rangle}\,.\nn
\ena
Due to this equation the left hand side happens to be 
up to the overall multiplier $\z^{\al}$ a rational function
of $\z^2$ with poles only
at $\z^2=q^{\pm 2}$.  
Its singular part is given as follows.
\begin{lem}\label{prop:charact1}
Set
\bea
&&\omega_{sing}(\z,\xi)=-\Delta_\z\psi(\z/\xi,\al)\label{omegasing}
\\
&&\quad+\frac{4}{T(\z,\kappa)T(\xi,\kappa)}
\left(a(\xi)d(q^{-1}\xi)\psi(q\z/\xi,\al)
-
a(q\xi)d(\xi)\psi(q^{-1}\z/\xi,\al)\right). 
\nn
\ena
Then we have
\begin{align}
&T(\z,\kappa)Z^\kappa
\Bigl\{\Bigl(
\bb^*(\z,\al)
-\frac 1 {2\pi i}\oint\limits_{\Gamma}\omega_{sing}(\z,\xi)
\cb(\xi,\al)\Bigr)(X)
\Bigr\}
\frac{d\xi^2}{\xi^2}=\z^\al P_\mathbf{n}(\z^2), \label{singb*}
\end{align}
where $X\in\mathcal{W}_{\al+1,-1}$, $\Gamma$ encircles $\xi^2=1$, and 
$P_\mathbf{n}(\z^2)$ is a polynomial in $\z^2$ of 
degree at most ${\bf n}$.
\end{lem}
 Lemma \ref{prop:charact1} is proved in Appendix 
\ref{app:bstar1}. 

In order to characterise the quantity in the 
left hand side of \eqref{singb*}, we 
need to have a control over the unknown polynomial
$P_\mathbf{n}(\z^2)$. This is the point where deformed Abelian integrals
come into play. 
Introduce the notation
\be
\overline{D}_\z F(\z)=F(q\z)+F(q^{-1}\z)
 -2\rho(\z)F(\z).
\en

\begin{lem}\label{prop:charact2}
For $\mathbf{m=0,\cdots, n}$, the following relations hold.
\begin{align}
&\int\limits_{\Gamma_{{\bf m}}}
T(\z,\kappa)
Z^\kappa\Bigl\{\Bigl(\bb^*(\z,\al)
+\frac 1 {2\pi i}\oint\limits_{\Gamma}\frac{d\xi^2}{ \xi^2}
\left(\overline{D}_\z\overline{D}_\xi\Delta^{-1}_{\z}
\psi(\z/\xi,\al)\right)
\cb(\xi,\al)\Bigr)(X)
\Bigr\}
\label{b*periods}
\\&\times 
 Q^-(\z,\kappa+\al)Q^+(\z,\kappa) \varphi(\z)   \frac{d\z^2}{\z^2}   =0,\nn
\end{align}
for $X\in\mathcal{W}_{\al+1,-1}$.
\end{lem}
As explained at the end of Section \ref{DAI}, 
one can apply Lemma \ref{1stAI} to 
$f^+(\z)=\overline{D}_\xi\psi(\z/\xi,\al)$. 
Then, the integral over $\z^2$ in the second term can 
be rewritten as 
\bea
&&\int\limits_{\Gamma_{{\bf m}}}
T(\z,\kappa)
\overline{D}_\z\overline{D}_\xi
\Delta^{-1}_\z\psi(\z/\xi,\al)Q^-(\z,\kappa+\al)Q^+(\z,\kappa)\varphi(\z)\frac{d\z^2}{\z^2}\label{trans}
\\
&&=\int\limits_{\Gamma_{{\bf m}}}
\overline{D}_\xi\psi(\z/\xi,\al)Q^-(\z,\kappa+\al)
\left(a(\z)Q^+(q^{-1}\z,\kappa)
-d(\z)Q^+(q\z,\kappa)\right)\varphi(\z)\frac{d\z^2}{\z^2}. \nn
\ena
Hence it does not actually depend on a particular choice 
of $\Delta^{-1}_\z\psi(\z/\xi,\al)$. 

Proof of Lemma \ref{prop:charact2} is long and technical.
We defer it to Appendix \ref{app:bstar2}. 

Comparing  \eqref{b*to c1} with \eqref{singb*}, \eqref{b*periods},
we infer that the function $\omega (\z,\xi)=\omega (\z,\xi|\kappa,\al)$ satisfy the conditions

\noindent 1. Singular part
\begin{align}
\z ^{-\al}T(\z,\kappa)\bigl(\omega(\z,\xi)-\omega_{sing}(\z,\xi)\bigr)\ 
\text{is a polynomial in}
\ \z ^2\ \text{ of degree }
\ \mathbf{n} \,.\label{singofom}
\end{align}

\noindent
2. Normalisation
\begin{align}
\int\limits_{\Gamma_{{\bf m}}}
T(\z,\kappa)
\left(\omega(\z,\xi)+
\overline{D}_\z\overline{D}_\xi
\Delta^{-1}_\z\psi(\z/\xi,\al)\right)Q^-(\z,\kappa+\al)Q^+(\z,\kappa)
\varphi(\z) \frac{d\z^2}{\z^2}=0\,,
\label{normofom}\\
\qquad\qquad\qquad\qquad\qquad\qquad\qquad\qquad
(\mathbf{m=0,\cdots,n})\,.
\end{align}
Furthermore, the equation \eqref{b*to c2} requires an additional
property of $\omega(\z,\xi|\kappa,\al)$ (see Section 8).

\noindent
3. Symmetry
\begin{align}
\omega(\xi,\z|-\kappa,-\al)=\omega(\z,\xi|\kappa,\al)\,.\label{symofom}
\end{align}



\section{Definition of $\omega (\z,\xi)$ and its symmetry}\label{defomega}

We shall first give the definition of the function $\omega (\z,\xi)$, 
and then prove that it satisfies all the necessary properties. 

In Section \ref{DAI}, we defined
the matrices $\mathcal{A}^+$ and $\mathcal{B}^+$. 
In Appendix \ref{nondeg} we show that the condition
\begin{align}
\det \mathcal{A}^+\ne 0\,.
\label{det}
\end{align}
is equivalent to the non-degeneracy condition 
(\ref{general}) accepted previously. 
The classical analogue of (\ref{det}) states that 
``there are no holomorphic differentials such that all the 
$a$-periods vanish".

Assuming (\ref{det}), consider the function
\begin{align}
&\omega (\z,\xi|\kappa,\al)=\frac 4 {T(\z,\kappa)T(\xi,\kappa)}
v^+(\z)^t(\mathcal{A}^+)^{-1}\mathcal{B}^+v^-(\xi)
+
\osym(\z,\xi|\kappa,\al)\,,
\label{def-omega2}
\end{align}
where 
$v^{\pm}(\z)$ are vectors with components
$v^{\pm}(\z)_{\mathbf{j}}=\z ^{\pm \al+2\mathbf{j}}$,
$\mathcal{A}$, $\mathcal{B}$ are given by \eqref{AB1},
\eqref{AB2},  and 
\begin{align}
\osym(\z,\xi|\kappa,\al)
=\frac 1 {T(\z,\kappa)T(\xi,\kappa)}\Bigl\{
&\(4a(\xi)d(\z)-T(\z,\kappa)T(\xi,\kappa)\)\psi(q\z/\xi,\al)
\nn\\
&-\(4a(\z)d(\xi)-T(\z,\kappa)T(\xi,\kappa)\)\psi(q^{-1}\z/\xi,\al)
\nn\\
-2\psi(\z/\xi,\al)\Bigl(&T(\z,\kappa)T(\xi, \kappa +\al)
 -T(\xi,\kappa)T(\z, \kappa +\al)\Bigr)
\Bigr\}\,.\nn
\end{align}

The function $\osym(\z,\xi|\kappa,\al)$ satisfies 
the relation
\begin{align}
\osym(\z,\xi|\kappa,\al)=\osym(\xi,\z|-\kappa,-\al)\label{osymm}
\end{align}
due to \eqref{symT} and the equality 
$\psi(\z^{-1},-\al)=-\psi(\z,\al)$.

The function $\z ^{-\al}\omega (\z,\xi|\kappa,\al)$ is a rational function of $\z^2$.
It is clear by the construction that 
the property (\ref{singofom}) is satisfied.

The remaining properties
(\ref{normofom}), (\ref{symofom})
of $\omega (\z,\xi|\kappa,\al)$
are more complicated, and we formulate them as lemmas.
\begin{lem}\label{lem:7.1}
The function $\omega (\z,\xi|\kappa,\al)$ defined 
by \eqref{def-omega2} satisfies the normalisation condition
\eqref{normofom}. 
\end{lem}
\begin{proof}
By using definitions \eqref{AB1} and \eqref{AB2}  we have,
\begin{align}
&\int\limits_{\Gamma _{\mathbf{m}}}
v^+(\z)^t(\mathcal{A}^+)^{-1}\mathcal{B}^+v^-(\xi)
Q^-(\z ,\kappa+\al)Q^+(\z,\kappa)\varphi (\z)\frac{d\z ^2}{\z ^2}
\nn\\
&=\left(\mathcal{B}^+v^-(\xi)\right)_{\mathbf{m}}
\nn\\
&=\int\limits _{\Gamma _{\mathbf{m}}}
r^+(\z,\xi)Q^-(\z ,\kappa+\al)Q^+(\z,\kappa)\varphi (\z)\frac{d\z ^2}{\z ^2}\,,
\nn
\end{align}
The definition \eqref{r+} can be rewritten as 
\begin{align}
r^+(\z,\xi)&=E\bigl(\psi(\z/\xi,\al)\bigr)
\nn\\
&-\frac{1}{4}T(\xi,\kappa)T(\z,\kappa)
\Bigl\{\omega_{sym}(\z,\xi|\kappa,\al)
+
\overline{D}_{\z}\overline{D}_{\xi}\Delta _{\z}^{-1}\psi (\z/\xi,\al)
\Bigr\}
\nn\,.
\end{align}
Therefore,  $T(\z,\kappa)\(\omega(\z,\xi)+
\overline{D}_\z\overline{D}_\xi
\Delta^{-1}_\z\psi(\z/\xi,\al)\)$ is a $q$-deformed exact form in $\z$.

\end{proof}
\begin{lem}\label{symmetry}
The function $\omega (\z,\xi|\kappa,\al)$ defined 
by \eqref{def-omega2} satisfies the  symmetry condition
\eqref{symofom}.
\end{lem}
\begin{proof}
In Section \ref{DAI} we had the relation (\ref{symplectic}), 
\begin{align}
\mathcal{B}^+(\mathcal{A}^-)^t=\mathcal{A}^+(\mathcal{B}^-)^t
\,.\label{:::}
\end{align}
In Appendix \ref{nondeg} we show that $\det (\mathcal{A}^-)\ne 0$
follows from the condition \eqref{general}. Hence both $\mathcal{A}^\pm$ can
be inverted. 
So, invert them and multiply the result by $v^+(\z)^t$ from the left and $v^-(\xi) $ from the right:
$$
v^+(\z)^t\(\mathcal{A}^+\)^{-1}
\mathcal{B}^+v^-(\xi)=v^-(\xi)^t\(\mathcal{A}^-\)^{-1}\mathcal{B}^-v^+(\z)
\,.
$$
What remains to do is to add $\osym(\z,\xi|\kappa,\al)$
to both sides, to use  (\ref{osymm}) 
and to recall the identities (\ref{symT}) and (\ref{symQ}).
\end{proof}
\section{Main theorem}\label{mainthm}

Now we are able to prove our main theorem.
\begin{thm}
Under 
the generality condition \eqref{general}
we have
\begin{align}
&Z^{\kappa}\bigl\{\bb^*(\z)(X)\bigr\}
=\frac 1 {2\pi i}\oint\limits_{\Gamma}
\omega (\z,\xi)
Z^{\kappa}\bigl\{\cb(\xi)(X)\bigr\}\frac{d\xi^2}{\xi^2}\,,
\label{mainb1}\\
&Z^{\kappa}\bigl\{\cb^*(\z)(X)\bigr\}=-\frac 1 {2\pi i}\oint\limits _{\Gamma}
\omega (\xi,\z)
Z^{\kappa}\bigl\{\bb(\xi)(X)\bigr\}\frac{d\xi^2}{\xi^2}\,.
\label{mainc1}
\end{align}
\end{thm}
\begin{proof}
Consider (\ref{mainb1}). It has been said that it is sufficient to consider the blocks
$\bb^*(\z,\al)$, $\cb(\xi,\al)$. 
Due to the structure of singularities (\ref{singb*}) and (\ref{singofom})
we have:
\begin{align}
T(\z,\kappa)
Z^{\kappa}\Bigl\{\Bigl(\bb^*(\z,\al)-
\frac 1 {2\pi i}\oint\limits _{\Gamma}
\omega (\z,\xi)
\cb(\xi,\al)\Bigr)(X)\Bigr\}\frac{d\xi^2}{\xi^2}
=\z ^{\al}
\tilde{P}_{\mathbf{n}}(\z^2)\,,
\end{align}
where $P_\mathbf{n}(\z^2)$ is a polynomial of degree $\mathbf{n}$. 
Due to Lemma \ref{prop:charact2}
and Lemma \ref{lem:7.1} we have
$$
\int\limits _{\Gamma _{\mathbf{m}}}\z ^\al P_{\mathbf{n}}(\z^2)Q^-(\z,\kappa+\al)
Q^+(\z,\kappa)\varphi(\z)\frac {d\z ^2}{\z ^2}=0,\quad \mathbf{m=0,\cdots, n}\,,
$$
which implies $P_{\mathbf{n}}(\z ^2)=0$ due to (\ref{det}).

Now consider (\ref{mainc1}). According to
\cite{HGSII}, the operators $\cb ^*$, $\bb$ 
are related to $\bb ^*$, $\cb$ 
by the transformation

\be
&&\phi_\al(\xb(\z,\al))=q^{-1}N(\al-1)\circ
\J\circ\xb(\z,-\al)\circ\J,
\en
where $N(x)=q^{-x}-q^{x}$ and $\J(X)=J X J^{-1}$ is the spin reversal. 
Namely,
$$
\cb^*(\z,\al)=-\phi_\al(\bb^*(\z,\al)),
\qquad\bb(\z,\al)=\phi_\al(\cb(\z,\al))\,.
$$
It is also easy to see that 
$$
Z^{\kappa}\{X\}=Z^{-\kappa}\{\J(X)\}\,.
$$
Hence (\ref{mainb1}) implies
\begin{align}
&Z^{\kappa}\bigl\{\cb^*(\z,\al)(X)\bigr\}=-\frac 1 {2\pi i}\oint\limits _{\Gamma}
\omega (\z,\xi|-\kappa,-\al)
Z^{\kappa}\bigl\{\bb(\xi,\al)(X)\bigr\}\frac{q\xi^2}{\xi^2}\,,\nn
\end{align}
which is equivalent to
(\ref{mainc1}) due to (\ref{symofom}).
\end{proof}

\appendix

\section{Proof of Lemma  \ref{prop:charact1}}
\label{app:bstar1}

In this appendix, we prove Lemma \ref{prop:charact1}. We
also prove 
some additional result used in Appendix B (Corollary \ref{corr}).
First let us  comment on the equation 
\eqref{redZb*}. 
This formula is used in the proof in order to
reduce the action of the operator $\bb^*$,  
when it is considered inside  the functional $Z^\kappa$,
to that of an operator on the interval $[1,m]$. 
This is a great simplification, because without $Z^\kappa$
the support of the coefficients in the expansion of $\bb^*(\z,\al)$ in
$\zeta^2-1$ becomes indefinitely large. 
In other words, inside $Z^\kappa$ the series expansion of $\bb^*(\z,\al)$
can be summed up to a rational function. Therefore,
the proof of Lemma \ref{prop:charact1}
consists in computing the singular part of the rational function.
This task is done indirectly by 
considering an inhomogeneous space chain.
We introduce inhomogeneity parameters $\xib=(\xi_j)$,
so that the original multiple poles $\zeta^2=q^{\pm2}$ in the homogeneous chain
are split into simple poles 
$\zeta^2=q^{\pm2}\xi_j^2$ for $1\leq j\leq m$.
Define the functional
\bea
&&Z^\kappa_{[1,m]}\left\{X_{[1,m]}\right\}
=
\frac{\Tr_{[1,m]}\langle\kappa+\al|T_{[1,m],{\bf M}}(\xib,\kappa)X_{[1,m]}
|\kappa\rangle}
{\prod_{j=1}^mT(\xi_j,\kappa)\langle\kappa+\al|\kappa\rangle}\,,
\label{Z-finite}
\\
&&T_{[1,m],{\bf M}}(\xib,\kappa)
=T_{1,{\bf M}}(\xi_1,\kappa)\cdots T_{m,{\bf M}}(\xi_m,\kappa). 
\nn
\ena
Using this functional the equation \eqref{redZb*} takes the form
\begin{align}
&Z^\kappa\left\{\bb^*(\z,\al)(q^{2(\al+1) S(0)}X_{[1,m]})\right\}
\label{specialization from inhomogeneous}
\\
&=Z^\kappa_{[1,m+1]}
\left\{
2\gb_{m+1,[1,m]}(\xi_{m+1},\al)(X_{[1,m]}) 
\right\}\mid_{\xi_1=\cdots=\xi_m=1,\xi_{m+1}=\zeta}\,.\nn
\end{align}

First recall from \cite{HGSII} the definition of the 
operator $\kb_{[1,m]}(\z,\al)$ and its basic relations with
$\cb_{[1,m]}(\z,\al)$, $\bar{\cb}_{[1,m]}(\z,\al)$,
$\fb_{[1,m]}(\z,\al)$:
\bea
&&\kb_{[1,m]}(\z,\al)(X_{[1,m]}) 
\label{kb}\\
&&\quad =\Tr_{a,A}\left\{\sigma^+_a\T_{\{a,A\},[1,m]}(\z,\al)
\z^{\al-\bS_{[1,m]}}
(q^{-2S_{[1,m]}}X_{[1,m]})\right\},
\nn\\
&&\kb_{[1,m]}(\z,\al)(X_{[1,m]})
-\Delta_\zeta\fb_{[1,m]}(\z,\al)(X_{[1,m]})
\label{decomp}\\
&&\quad 
=
\cb_{[1,m]}(q\z,\al)(X_{[1,m]})
+
\cb_{[1,m]}(q^{-1}\z,\al)(X_{[1,m]})
+
\bar\cb_{[1,m]}(\z,\al)(X_{[1,m]}). 
\nn
\ena
in the complex plane the operators $\cb_{[1,m]}(\z,\al)$, $\bar{\cb}_{[1,m]}(\z,\al)$,
$\fb_{[1,m]}(\z,\al)$ have singularities at $\z ^2=\xi _j^2$ only.
The operator  $\gb_{c,[1,m]}(\z,\al)$ is given by
\bea
&&2\gb_{c,[1,m]}(\z,\al)(X_{[1,m]})
=\fb_{[1,m]}(q\z,\al)(X_{[1,m]})
+\fb_{[1,m]}(q^{-1}\z,\al)(X_{[1,m]})
\label{gc}\\
&&-2\bT_{c,[1,m]}(\z,\al)\fb_{[1,m]}(\z,\al)(X_{[1,m]})
+2\ub_{c,[1,m]}(\z,\al)(X_{[1,m]}),
\nn
\ena
where 
\be
&&\ub_{c,[1,m]}(\z,\al)(X_{[1,m]})
=\Tr_{c,a,A}\left(Y_{a,c,A}
\bT_{\{a,A\},[1,m]}(\z,\al)\z^{\al-\bS_{[1,m]}}
(q^{-2S_{[1,m]}}X_{[1,m]})\right),
\nn\\
&&Y_{a,c,A}=
-\frac{1}{2}\sigma^3_c\sigma^+_a
+\sigma^+_c\sigma^3_a-\ao_A\sigma^+_c\sigma^3_a. 
\nn
\en

For the proof of Lemma \ref{prop:charact1}  we compare the singularities of
$\gb_{c,[1,m]}(\zeta,\al)(X_{[1,m]})$ inside the functional 
$Z^\kappa_{[1,m+1]}$ with those of
$\cb_{[1,m]}(\z,\al)(X_{[1,m]})$ inside $Z^\kappa_{[1,m]}$.

It is known that $\gb_{c,[1,m]}(\zeta,\al)(X_{[1,m]})$ is regular at $\zeta^2=\xi_j^2$. 
Now we compare $\res_{\z^2=q^{\pm2}\xi_m^2}\gb_{c,[1,m]}(\z,\al)(X_{[1,m]})$ with $\res_{\z^2=\xi_m^2}\cb_{[1,m]}(\z,\al)(X_{[1,m]})$.
Set 
\be
&&U_{[1,m]}=\res_{\z^2=\xi_m^2}\qb_{[1,m]}(\z,\al)(X_{[1,m]})
\frac{d\z^2}{\z^2},
\\
&&\qb_{[1,m]}(\z,\al)(X_{[1,m]})=
\Tr_{A}\left(\bT_{A,[1,m]}(\z,\al)
\z^{\al-\bS_{[1,m]}}(q^{-2S_{[1,m]}}X_{[1,m]})\right).
\en
\begin{lem}\label{lem:rescgc}
The operator $[\sigma^+_m,U_{[1,m]}]_+$ has its support in 
$[1,m-1]$$\mathrm{:}$
\bea
[\sigma^+_m,U_{[1,m]}]_+=x_{[1,m-1]}I_m,
\quad x_{[1,m-1]}=\Tr_m\left(\sigma^+_mU_{[1,m]}\right)\,. 
\label{support}
\ena
We have 
\begin{align}
&\res_{\z^2=\xi_m^2}\cb_{[1,m]}(\z,\al)(X_{[1,m]})
\frac{d\z^2}{\z^2}
=-\frac{1}{2}[\sigma^+_m,U_{[1,m]}]_+\,,
\label{cres}\\
&\res_{\z^2=q^{2\varepsilon}\xi_m^2}
\left(\gb_{c,[1,m]}(\z,\al)(X_{[1,m]})+\bT_{c,[1,m]}(\z,\al)
\fb_{[1,m]}(\z,\al)(X_{[1,m]})\right)\frac{d\z^2}{\z^2}
\label{gcres}\\
&\quad =
\begin{cases}
-\frac{1}{4}[\sigma^+_m,U_{[1,m]}]_+
-U_{[1,m]}(\tau^+_m\sigma^+_c-\sigma^+_m\tau^+_c) & (\varepsilon=+);\\ 
\frac{1}{4}[\sigma^+_m,U_{[1,m]}]_+
+(\tau^-_m\sigma^+_c-\sigma^+_m\tau^-_c)U_{[1,m]} & 
(\varepsilon=-)\,.\\ 
\end{cases} \nn
\end{align}
\end{lem}
\begin{proof}
Property \eqref{support} appears in \cite{FB} as 
Lemma 2.6 (see also \cite{HGSII}, Appendix D). 
Formula \eqref{cres} is proved in \cite{HGSII}, Lemma 2.2.  
The calculation for \eqref{gcres} is similar, but
we omit the details. 
\end{proof}
\begin{cor}\label{corr}
We have the  relations between 
 $\bar\cb_{[1,m]}(\z,\al) $ and
$ \cb_{[1,m]}(\z,\al)$
\begin{align}
&\res_{\z^2=\xi_j^2}\bigl(\bar\cb_{[1,m]}(\z,\al)+\tb^*_{[1,m]}(\z,\al)
\cb_{[1,m]}(\z,\al)\bigr)(X_{[1,m]})
\frac{d\z^2}{\z^2}=0
\,.
\label{cbarres}
\end{align}
\end{cor}
\begin{proof}
To see \eqref{cbarres}, it suffices to write 
\begin{align}
\res_{\z^2=\xi_m^2}&\bar\cb_{[1,m]}(\z,\al)(X_{[1,m]}) 
=\Tr_{a}\left(\sigma^+_a\mathbb{P}_{a,m}
\bT_{a,[1,m-1]}(\xi_m,\al)U_{[1,m]}\right)\nn\\&
=
\Tr_{a}\Bigl(\mathbb{P}_{a,m}
\bT_{a,[1,m-1]}(\xi_m,\al)\Bigl(
(0,1)_mU_{[1,m]}
\Bigl(\begin{matrix}1\\ 0\\ \end{matrix}\Bigr)_m\ \Bigr)\Bigr)\,,\nn
\end{align}
and use \eqref{support}, \eqref{cres} and $R$-matrix symmetry.
\end{proof}

Using Lemma \ref{lem:rescgc} we obtain
\begin{align}
&\res_{\xi_{m+1}^2=q^{2}\xi_m^2}
Z^\kappa_{[1,m+1]}
\left\{2\gb_{m+1,[1,m]}(\xi_{m+1},\al)(X_{[1,m]})\right\}
\frac{d\xi_{m+1}^2}{\xi_{m+1}^2}
\label{xim2}\\
&\quad=
\res_{\zeta^2=\xi_m^2}
Z^\kappa_{[1,m]}\left\{\cb_{[1,m]}(\zeta,\al)(X_{[1,m]})\right\}
\frac{d\zeta^2}{\zeta^2}\nn\\
&\quad-2\ \res_{\xi_{m+1}^2=q^{2}\xi_m^2}Z^\kappa_{[1,m+1]}
\left\{U_{[1,m]}(\tau^+_m\sigma^+_{m+1}
-\sigma^+_m\tau^+_{m+1})\right\}.
\nn
\end{align}
Here the third term of \eqref{gc} does not contribute, because the only singularities of $\fb_{[1,m]}(\z,\al)$ are the simple poles at $\z^2=\xi_j^2$, and
the following inhomogeneous analogue of Theorem \ref{thm:t*rho} holds:
\bea
Z^\kappa_{[1,m+1]}
\left\{\bT_{m+1,[1,m]}(\xi_{m+1},\al)X_{[1,m]}\right\}
=2\rho(\xi_{m+1})Z^\kappa_{[1,m]}\left\{X_{[1,m]}\right\}. 
\label{Zt*rho}
\ena
Note that 
\be
&&\tau^+_m\sigma^+_{m+1}-\sigma^+_m\tau^+_{m+1}
=(\tau^+_m\sigma^+_{m+1}-\sigma^+_m\tau^+_{m+1})P^-_{m,m+1},
\en
where $P^-$ is the projector on the singlet, 
$$
P^-(v_\varepsilon\otimes v_{\varepsilon'})=
\epsilon\delta_{\epsilon+\epsilon',0}
\frac{1}{2}(v_+\otimes v_- -v_-\otimes v_+),
\quad
\sigma^3v_\varepsilon=\varepsilon v_\varepsilon.
$$
Using the cyclicity of trace and the quantum determinant relation
\be
&&
P^-_{m,m+1}
T_{m,{\bf M}}(\xi_m)T_{m+1,{\bf M}}(q\xi_m)
=a(q\xi_m)d(\xi_m)P^-_{m,m+1}, 
\en
we find 
\begin{align*}
&\res_{\xi_{m+1}^2=q^{2}\xi_m^2}Z^\kappa_{[1,m+1]}
\left\{U_{[1,m]}(\tau^+_m\sigma^+_{m+1}
-\sigma^+_m\tau^+_{m+1})\right\}=\frac{a(q\xi_m)d(\xi_m)}{\prod_{j=1}^mT(\xi_j,\kappa)\cdot T(q\xi_m,\kappa)}\nn
\\
&\times
\frac{\Tr_{[1,m+1]}\langle\kappa+\al|T_{[1,m-1],{\bf M}}(\xib,\kappa)
U_{[1,m]}(\tau^+_m\sigma^+_{m+1}-\sigma^+_m\tau^+_{m+1})|\kappa\rangle}
{\langle\kappa+\al|\kappa\rangle}\\
&=-\frac{a(q\xi_m)d(\xi_m)}{T(\xi_m,\kappa) T(q\xi_m,\kappa)}
Z^\kappa_{[1,m-1]}\Bigl\{\Tr_m\frac{1}{2}[\sigma^+_m, U_{[1,m]}]_+\Bigr\}\\
&=\frac{2a(q\xi_m)d(\xi_m)}
{T(\xi_m,\kappa)T(q\xi_m,\kappa)}
\res_{\zeta^2=\zeta_m^2}
Z_{[1,m]}^\kappa
\left\{\cb_{[1,m]}(\zeta,\al)(X_{[1,m]})\right\}
\frac{d\zeta^2}{\zeta^2}.
\end{align*}
In the last line we used \eqref{support}, \eqref{cres}.
Computation of the 
residue at $\z^2=q^{-2}\xi_m^2$ is done similarly, using
\be
&&\tau^-_m\sigma^+_{m+1}-\sigma^+_m\tau^-_{m+1}
=P^-_{m,m+1}(\tau^-_m\sigma^+_{m+1}-\sigma^+_m\tau^-_{m+1}).
\en
The residues at $\xi_{m+1}^2=q^{\pm 2}\xi_j^2$ are 
readily found from $R$-matrix symmetry. 
\begin{lem}
\begin{align}
&
\res_{\xi _{m+1}^2=q^{\pm 2}\xi_j^2}
Z^\kappa_{[1,m+1]}\left\{2\gb_{m+1,[1,m]}(\xi_{m+1},\al)(X_{[1,m]}) 
\right\}
\frac{d\xi_{m+1}^2}{\xi_{m+1}^2}\label{rescb1}
\\
&\quad=\res_{\z^2=q^{\pm 2}\xi_j^2}\ \omega (\z,\xi _j)
\ \res_{\zeta^2=\xi_j^2}
Z^\kappa_{[1,m]}
\left\{
\cb_{[1,m]}(\zeta,\al)(X_{[1,m]})
\right\}
\frac{d\zeta^2}{\zeta^2}.
\nn
\end{align}
\end{lem}
\begin{proof}
This follows from the preceding calculations and 
\be
&&
\res_{\z^2=q^{ 2}\xi_j^2}\ \omega (\z,\xi _j)=1-\frac{4a(\xi_j)d(q\xi_j)}{T(\xi _j,\kappa)T(\xi _j q,\kappa)}\,,\\
&&
\res_{\z^2=q^{ -2}\xi_j^2}\ \omega (\z,\xi _j)=
-\Big(1-\frac{4a(\xi_j q^{-1})d(\xi_j)}{T(\xi _j,\kappa)T(\xi _j q^{-1},\kappa)}\Bigr)\,.
\en
\end{proof}
Let us return to the homogeneous case $\xi_1=\cdots=\xi_m=1$.
The operators $\cb(\z,\al)$, $\bar{\cb}(\z,\al)$ acting from  $\mathcal{W}_{\al+1,-1}$ to
$\mathcal{W}_{\al,0}$ are defined by
\begin{align}
&\cb(\z,\al)\bigl(q^{2(\al+1)S(0)}X_{[1,m]}\bigr)=q^{2\al S(0)}\cb_{[1,m]}(\z,\al)\bigl(X_{[1,m]}\bigr)\,,\label{defc}\\
&\bar{\cb}(\z,\al)\bigl(q^{2(\al+1)S(0)}X_{[1,m]}\bigr)=q^{2\al S(0)}\bar{\cb}_{[1,m]}(\z,\al)\bigl(X_{[1,m]}\bigr)\,,\nn
\end{align}
and the requirement of translational invariance. This definition is equivalent to the one given in \cite{HGSII}. The reduction relation 
proven 
there ensures the self consistency of the present definition. 
The equation (\ref{singb*}) follows by writing \eqref{rescb1}
as a contour integral and 
specialising to $\xi_1=\cdots=\xi_m=1$.

It remains to show that the polynomial 
$P_\mathbf{n}(\z ^2)$ in the remainder term of \eqref{singb*} 
has degree at most $\mathbf{n}$.
The only non-trivial case to consider 
is when the spin of $X_{[1,m]}$ equals $-1$.
Then it follows from the fact \cite{HGSII} that 
$$
\gb_{c,[1,m]}(\z,\al)(X_{[1,m]})=O(1),\qquad \z^2\to \infty\,.
$$
This finishes the proof of Lemma \ref{prop:charact1} .


\section{Proof of Lemma \ref{prop:charact2}}
\label{app:bstar2}
The goal of this section is to prove Lemma \ref{prop:charact2}.
The proof is done in several steps. 

\noindent{\bf Step 1.}

Recall the definition \eqref{kb}. 
Fix a solution $\fb_{0,[1,m]}(\z,\al)$ of the equation
\bea
\Delta_\z\fb_{0,[1,m]}(\z,\al)(X_{[1,m]})
=\kb_{[1,m]}(\z,\al)(X_{[1,m]}) 
\label{delinv-k}
\ena
which has poles only at $\z^2=q^{2n}$ ($n\in \Z$). 
Define further
\begin{align}
&\bb^*_0(\z,\alpha)(q^{2(\alpha+1)S(0)}X_{[1,m]})\label{bbstar0}\\
&\qquad\qquad =
\lim_{l\rightarrow\infty}q^{2\alpha S(0)}\Tr_c\{\mathbb T_{c,[m+1,l]}(\z)\gb_{0,c,[1,m]}(\z,\alpha)(X_{[1,m]})\},\nn
\\
&\gb_{0,c,[1,m]}(\z,\alpha)(X_{[1,m]})\label{gb*0}\\
&\qquad\qquad=\frac12\fb_{0,[1,m]}(q\z,\al)(X_{[1,m]})
+\frac12\fb_{0,[1,m]}(q^{-1}\z,\al)(X_{[1,m]})\nn
\\
&\qquad\qquad-\mathbb T_{c,[1,m]}(\z,\alpha)\fb_{0,[1,m]}(\z,\al)(X_{[1,m]})+\ub_{c,[1,m]}(\z,\al)(X_{[1,m]})\,.
\nn
\end{align}
\begin{lem}\label{lem:step2}
Define
\begin{align}
&I_{\mathbf{M}}(\z)(X_{[1,m]})\nn\\
&=Q_{\mathbf{M}}^-(\z,\kappa+\al)
\Tr _{[1,m],c}\Bigl(T_{[1,m],\mathbf{M}}(1,\kappa)T_{c,\mathbf{M}}(\z,\kappa)\gb_{0,c,[1,m]}(\z,\al)(X_{[1,m]})\Bigr)Q^+_{\mathbf{M}}(\z,\kappa)\,.\nn
\end{align}
The identity \eqref{b*periods} follows from 
\begin{align}
\int\limits_{\Gamma_{{\bf m}}}I_{\mathbf{M}}(\z)(X_{[1,m]})
\varphi(\z)\frac{d\z^2}{\z^2}=0\,.\label{step2}
\end{align}
\end{lem}
\begin{proof}

Introduce the operator
\be
D_\z F(\z)=F(q\z)+F(q^{-1}\z)-\tb^*(\z)F(\z),
\en
which can be used to rewrite the definition of $\bb ^*(\z,\al)$:
\begin{align}
\bb ^*(\z,\al)(q^{2(\al +1)S(0)}X_{[1,m]})&=D_\z
\(q^{2\al S(0)}\fb_{[1,m]} ^*(\z,\al)(X_{[1,m]})\)\nn\\ &+
q^{2\al S(0)}\lim_{l\rightarrow\infty}\Tr_c\{\mathbb T_{c,[m+1,l]}(\z)\ub_{c,[1,m]}(\z,\alpha)(X_{[1,m]})\}\,.\nn
\end{align}
In similar formula for $\bb _0^*(\z,\al)$ the only change is $\fb$
to $\fb _0$, $\ub$ remains the same.
Comparing this with the equations \eqref{decomp},
\eqref{delinv-k} and \eqref{defc} we arrive at
$$\bb^*_0(\z,\al)-\bb ^*(\z,\al)=D_\z\Delta _\z^{-1}
\(\cb(\z q,\al)+\cb(\z q^{-1},\al)+\bar{\cb}(\z,\al)\)\,.$$

Now consider the term containing $\cb$ in \eqref{b*periods}.
Notice that due to Lemma \ref{thm:t*rho} for any quasi-local operator
$X(\z)$ depending on $\z$
$$Z^\kappa\bigl\{\overline{D}_\z(X(\z))  \bigl\}=Z^\kappa\bigl\{D_\z(X(\z))  \bigl\}\,.$$
So we replace in \eqref{b*periods} $\overline{D}_\z$, $\overline{D}_\xi$ by
$D_\z$, $D_\xi$.
On the other hand from \eqref{cbarres} 
it follows that we have
an equality of formal power series in $(\z^2-1)^{-1}$, 
\begin{align}
\bar{\cb}(\z,\al)&(q^{2(\alpha+1)S(0)}X_{[1,m]})\label{step1-3}\\
&=-
\frac 1 {2\pi i}\oint\limits_{\Gamma}
\psi (\z/\xi,\al)\tb^*(\xi,\alpha)\cb(\xi,\alpha)(q^{2(\alpha+1)S(0)}X_{[1,m]})\frac{d\xi^2}{ \xi^2} \,.\nn
\end{align}

Using \eqref{step1-3} we evaluate
$$
\cb(\z q,\al)+\cb(\z q^{-1},\al)+\bar{\cb}(\z,\al)=\frac 1 {2\pi i}\oint\limits_{\Gamma}D_\xi\(
\psi(\z/\xi,\al)\)
\cb(\xi,\al)\frac{d\xi^2}{\xi ^2}\,.$$ 
In other words we obtain an equation similar to \eqref{redZb*}
with $\bb ^*$ replaced by the expression under $Z^{\kappa}$ in \eqref{b*periods}, and $\gb$ replaced by $\gb_0$,
which is nothing but the matrix element of \eqref{step2}.

\end{proof}

\noindent{\bf Step 2.}\quad

The next step is to reduce the identity to a difference equation for $\gb_0$ on a finite interval.
We will show that, for all ${\bf m=1,\cdots, n}$,
the identity \eqref{step2} reduces to the same equation
for a quantity in the space direction. 
So, we can forget the Matsubara direction. 
Introduce an operator
\be
\bbA_{c,[1,m]}(\z)(Y_{[1,m]\sqcup c})
=T_{c,[1,m]}(\z)q^{\al\sigma^3_c}
\theta_c\left(Y_{[1,m]\sqcup c}\ 
\theta_c\bigl(T_{c,[1,m]}(\z)^{-1}\bigr)
\right), 
\en
where $\theta$ signifies the anti-involution
\be
\theta(x)=\sigma^2 x^{t} \sigma^2\qquad 
(x\in \End(V)). 
\en

\begin{lem}\label{lem:step3}
Identity \eqref{step2} follows from the equation
\begin{align}
\gb_{0,c,[1,m]}(\z,\al)(X_{[1,m]})
&=-\bbA_{c,[1,m]}(\z)\left(
\gb_{0,c,[1,m]}(q^{-1}\z,\al)(X_{[1,m]})
\right)\,.
\label{step3}
\end{align}
\end{lem}
\begin{proof}
By symmetry it suffices to consider the case ${\bf m}={\bf n}$.
We prove the assertion assuming that ${\bf s}_{\bf n}=1/2$. 
The general case is reduced to this case by 
the standard fusion procedure.

From the defining relation \eqref{eqphi}
for $\varphi(\z)$, we have
\be
\res_{\z^2=q^{-2}\tau_\mathbf{n}^2}\varphi(\z)\frac{d\z^2}{\z^2}
=\frac{a(\tau_\mathbf{n})}{d(q^{-1}\tau_\mathbf{n})}
\res_{\z^2=\tau_\mathbf{n}^2}\varphi(\z)\frac{d\z^2}{\z^2}\,.
\en
So, the equation (\ref{step2}) is equivalent to:
\begin{align}
d(\tau _\mathbf{n}q^{-1})\ I_{\mathbf{M}}
(\tau_\mathbf{n})(X_{[1,m]})+a(\tau _\mathbf{n})\ I_{\mathbf{M}}
(\tau_\mathbf{n}q^{-1})
(X_{[1,m]})=0\,. \label{I+I}
\end{align}
Let us compute $I_{\mathbf{M}}(\tau _{\mathbf{n}})(X_{[1,m]})$. We simplify notations introducing
$$Y_{[1,m]}(\z,\al)=\gb_{0,c,[1,m]}(\z,\al)(X_{[1,m]})\,.$$
First, move $T_{c,\mathbf{M}}(\tau _\mathbf{n})$ to the left using the Yang-Baxter equation:
\begin{align}
&\Tr _{[1,m],c}\Bigl(T_{[1,m],\mathbf{M}}(1,\kappa)T_{c,\mathbf{M}}(\tau _\mathbf{n},\kappa)
Y_{[1,m]}(\tau_\mathbf{n},\al)\Bigr)\nn\\
&=\Tr _{[1,m],c}\Bigl(T_{c,\mathbf{M}}(\tau _\mathbf{n},\kappa)T_{[1,m],\mathbf{M}}(1,\kappa)
T_{c,[1,m]}(\tau _\mathbf{n})^{-1}
Y_{[1,m]}(\tau_\mathbf{n},\al)T_{c,[1,m]}(\tau _\mathbf{n})\Bigr)\,.\nn
\end{align}
Now, for $s_\mathbf{n}=1/2$, the $L$ operator satisfies
\be
&&L_{c,{\bf n}}(1)=\eta P_{c,{\bf n}},
\en
where we have set $\eta=q^{1/2}(q-q^{-1})$. 
Note that 
\be
&&T_{c,{\bf M}}(\tau_{\bf n},\kappa)=
\eta 
P_{c,{\bf n}}T_{c,{\bf M}'}(\tau_{\bf n},\kappa)
=\eta T_{\mathbf{n},{\bf M}'}(\tau_{\bf n},\kappa)P_{c,{\bf n}}\,,
\\
&&
\eta T_{{\bf n},{\bf M}'}(\tau_{\bf n},\kappa)
=T_\mathbf{M}(\tau_\mathbf{n},\kappa+\al)q^{-\al\sigma^3_{\bf n}}\,,
\en
where $\mathbf{M}'$ signifies the subinterval 
$[\mathbf{1},\mathbf{n-1}]$. 
Moreover,
$$T_\mathbf{M}(\tau _\mathbf{n},\kappa+\al)
Q_{\mathbf{M}}^-(\tau_\mathbf{n},\kappa+\al)=
a(\tau_\mathbf{n})Q_{\mathbf{M}}^-(\tau_\mathbf{n}q^{-1},\kappa+\al)\,,$$
because $d(\tau_\mathbf{n})=0$. So we can evaluate $I_{\mathbf{M}}(\tau _\mathbf{n})(X_{[1,m]})$ as 
\begin{align}
&I_{\mathbf{M}}(\tau _\mathbf{n})(X_{[1,m]})=a(\tau_\mathbf{n})Q_{\mathbf{M}}^-(\tau_\mathbf{n}q^{-1},\kappa+\al)\nn\\&\times
\Tr _{[1,m],c}\Bigl(
P_{c,{\bf n}}
T_{[1,m],\mathbf{M}}(1,\kappa)
q^{-\al\sigma^3_{c}}
T_{c,[1,m]}(\tau_\mathbf{n})^{-1}Y_{[1,m]}(\tau_\mathbf{n},\al)T_{c,[1,m]}(\tau_\mathbf{n})\Bigr)Q^+_{\mathbf{M}}(\tau_\mathbf{n},\kappa)\,.\nn
\end{align}
Now notice that
$$T_{[1,m],\mathbf{M}}(1,\kappa)=T_{[1,m],\mathbf{n}}(\tau _\mathbf{n}^{-1})T_{[1,m],\mathbf{M}'}(1,\kappa)=\mu (\tau _\mathbf{n})
T_{\mathbf{n},[1,m]}(\tau _\mathbf{n})^{-1}T_{[1,m],\mathbf{M}'}(1,\kappa)
\,,$$
where $\mu(\tau)$ is a function whose explicit form is irrelevant for
our calculation.

Bring the permutation through $T_{[1,m],\mathbf{M}}$ and 
put
$T_{[1,m],c}$ to the right by cyclicity of trace:
\begin{align}
I_{\mathbf{M}}(\tau _\mathbf{n})(X_{[1,m]})&=\mu(\tau_\mathbf{n})a(\tau_\mathbf{n})
Q_{\mathbf{M}}^-(\tau_\mathbf{n}q^{-1},\kappa+\al)
\nn\\&\times
\Tr _{[1,m],c}\Bigl(
T_{[1,m],\mathbf{M}'}(1,\kappa)P_{c,{\bf n}}
q^{-\al\sigma^3_{c}}
T_{c,[1,m]}(\tau_\mathbf{n})^{-1}Y_{[1,m]}(\tau_\mathbf{n},\al)\Bigr)Q^+_{\mathbf{M}}(\tau_\mathbf{n},\kappa)\,.\nn
\end{align}

We compute $I_{\mathbf{M}}(\tau _\mathbf{n}q^{-1})(X_{[1,m]})$ similarly, using 
\begin{align}
&T_{c,\mathbf{M}}(\tau _\mathbf{n}q^{-1},\kappa)=-2\eta q^{-1}
P^-_{c,\mathbf{n}}
T_{c,\mathbf{M}'}(\tau _\mathbf{n}q^{-1},\kappa)=
2P^-_{c,\mathbf{n}}T_{\mathbf{M}}(\tau _\mathbf{n}q^{-1}
\kappa)
\,,\nn\\
&T_{\mathbf{M}}(\tau _\mathbf{n}q^{-1},
\kappa)Q^+_{\mathbf{M}}(\tau _\mathbf{n}q^{-1},
\kappa)=d(\tau _\mathbf{n}q^{-1})Q^+_{\mathbf{M}}(\tau _\mathbf{n},
\kappa)\,.\nn
\end{align}
The result is
\begin{align}
&I_{\mathbf{M}}(\tau _\mathbf{n}q^{-1})(X_{[1,m]})=\mu(\tau_\mathbf{n})d(\tau _\mathbf{m}q^{-1})
Q_{\mathbf{M}}^-(\tau_\mathbf{n}q^{-1},\kappa+\al)
\nn\\&\times
\Tr _{[1,m],c}\Bigl(
T_{[1,m],\mathbf{M}'}(1,\kappa)P_{c,{\bf n}}\theta _c\bigl(Y_{[1,m]}(\tau_\mathbf{n}q^{-1},\al)
\theta _c\(T_{c,[1,m]}(\tau_\mathbf{n})^{-1}\)\bigr)\Bigr)Q^+_{\mathbf{M}}(\tau_\mathbf{n},\kappa)\,,\nn
\end{align}
where we applied the anti-automorphism $\theta _c$ under
$\Tr _c$ using $\Tr (\theta(x))=\Tr (x)$ and
$\theta _c(2P^-_{c,\mathbf{n}})=P_{c,\mathbf{n}}$. 
Thus (\ref{step3}) implies (\ref{I+I}).
\end{proof}


\noindent{\bf Step 3.}\quad

The third step is to reduce \eqref{step3} to representation 
theory. 

\begin{lem}\label{lem:step4}
Set
\be
&&y_1=\tau^-_c\sigma^+_a-\tau^-_a\sigma^+_c\,,
\\
&&
y_2=
\sigma^-_c\sigma^+_a
+\tau^-_c\tau^+_a-\tau^+_c\tau^-_a
+(\tau^-_c\sigma^+_a-\sigma^+_c\tau^+_a)\ao_A
 -\sigma^+_c\sigma^+_a\ao_A^2\,.
\en
Then equation \eqref{step3} is equivalent to the identities
\bea
&&\Tr_{c,a,A}\left(y 
\bT_{c,[1,m]}(\z,\al)
\bT_{a,[1,m]}(q^{-1}\z,\al)\bT_{A,[1,m]}(q^{-1}\z,\al)
(X_{[1,m]}')\right)=0
\qquad\label{trk}
\\
&&
(y=y_1,y_2),\nn
\ena
where we have set $X_{[1,m]}'=
(q^{-1}\z)^{\al-\bS_{[1,m]}}q^{-2S_{[1,m]}}X_{[1,m]}$. 
\end{lem}
\begin{proof}
Recall the definitions \eqref{delinv-k}, \eqref{gb*0}.
From the definition of $\bbA_{c,[1,m]}$ we easily find that 
\bea
&&\bbA_{c,[1,m]}(\z)(x_cX_{[1,m]})
=\bT_{c,[1,m]}(\z,\al)(X_{[1,m]})\theta_c(x_c)\,,
\label{bbA1}
\\
&&\bbA_{c,[1,m]}(\z)\bT_{c,[1,m]}(q^{-1}\z,\al)(X_{[1,m]})
=X_{[1,m]}\,.
\label{bbA2}
\ena
Write 
\be
\bu_{c,[1,m]}(\z,\al)(X_{[1,m]})
=-\frac{1}{2}\sigma^3_c\bk_{[1,m]}(\z,\al)(X_{[1,m]})
+\sigma^+_c\bl_{[1,m]}(\z,\al)(X_{[1,m]})\,. 
\en
Applying \eqref{bbA1}, \eqref{bbA2} 
we reduce \eqref{step3} to the form
\be
0&=&\left(
\bk_{[1,m]}(\z,\al)(X_{[1,m]})
 -\bT_{c,[1,m]}(\z,\al)(\bk_{[1,m]}(q^{-1}\z,\al)(X_{[1,m]}))
\right)\tau^-_c
\\
&+&
\left(
\bl_{[1,m]}(\z,\al)(X_{[1,m]})-
\bT_{c,[1,m]}(\z,\al)(\bl_{[1,m]}(q^{-1}\z,\al)(X_{[1,m]}))
\right)\sigma^+_c. 
\en
We rewrite this further,  using the fact that 
$y_c=0$ if and only if $\Tr_c(x_cy_c)=0$ for any $x_c\in \End(V_c)$. 
Nontrivial conditions arise from the choices 
$x_c=\tau^-_c$ or $\sigma^-_c$, giving respectively 
\bea
&&\bk_{[1,m]}(\z,\al)(X_{[1,m]})
=
\Tr_{c,a,A}\Bigl((\sigma^+_a\tau^-_c+(\sigma^3_a-\ao_A\sigma^+_a)\sigma^+_c)
\label{id3}
\\
&&\qquad \times
\bT_{c,[1,m]}(\z,\al)\bT_{\{a,A\},[1,m]}
(q^{-1}\z,\al)(q^{-1}\z)^{\al-\bS_{[1,m]}}
(q^{-2S_{[1,m]}}X_{[1,m]})\Bigr),\nn
\\
&&\bl_{[1,m]}(\z,\al)(X_{[1,m]})=
\Tr_{c,a,A}\Bigl(
(\sigma^+_a\sigma^-_c+(\sigma^3_a-\ao_A\sigma^+_a)\tau^+_c)
\label{id4}\\
&&\qquad \times
\bT_{c,[1,m]}(\z,\al)\bT_{\{a,A\},[1,m]}(q^{-1}\z,\al)
(q^{-1}\z)^{\al-\bS_{[1,m]}}(q^{-2S_{[1,m]}}X_{[1,m]})\Bigr).
\nn
\ena
On the other hand, we have an identity (see \cite{HGSII}, (2.20))
\be
\bT_{A,[1,m]}(\z,\al)(X_{[1,m]})
=\Tr_a\left(\tau^+_a \bT_{\{a,A\},[1,m]}(q^{-1}\z,\al)
(q^{-1}\z)^{\al-\bS}(X_{[1,m]})\right)\,, 
\en
which allows us to rewrite the left hand side of \eqref{id3} as 
\be
&&\bk_{[1,m]}(\z,\al)(X_{[1,m]})\\
&&\quad=
\Tr_{c,a,A}\left(
\sigma^+_c\tau^+_a
\bT_{c,[1,m]}(\z,\al)
\bT_{\{a,A\},[1,m]}(q^{-1}\z,\al)(q^{-1}\z)^{\al-\bS_{[1,m]}}
(q^{-2S_{[1,m]}}X_{[1,m]})\right)\,.
\en
For $\bl_{[1,m]}(\z,\al)(X_{[1,m]})$, we obtain an analogous expression,
replacing $\sigma^+_c\tau^+_a$ by \newline
$(\sigma^3_c+\ao_A\sigma^+_c)\tau^+_a$.

After this rewriting we take the difference of the left and the right hand sides 
of \eqref{id3}, \eqref{id4},  
and do further the gauge transformation 
\be
F_{a,A}\cdot \bT_{\{a,A\},[1,m]}(\z,\al)\cdot F_{a,A}^{-1}
=\bT_{a,[1,m]}(\z,\al)\bT_{A,[1,m]}(\z,\al)
\qquad (F_{a,A}=1-\ao_A\sigma^+_a)\,. 
\en
The assertion of Lemma follows. 
\end{proof}

To finish, it remains to prove \eqref{trk}. 
Let $\mathcal{R}$ be the universal $R$ matrix of 
$U_q(\widehat{\mathfrak{sl}}_2)$. 
Denote by $\pi_\z$ the evaluation module over $V=\C^2$, 
and by $\varpi_\z$ that of $q$-oscillator representation 
$W$. 
For the notation and details, 
we refer to Appendix A of \cite{HGSII}. 
Let further $\pi_{[1,m]}=\pi^{\otimes m}_1$. 
The identities of Lemma \ref{lem:step4}
can be written as 
\bea
&&\Tr_{c,a,A}\left(y\ 
(\pi_{\z}\otimes \pi_{q^{-1}\z}
\otimes \varpi_{q^{-1}\z}\otimes \pi_{[1,m]})\mathcal{R}\right)=0
\qquad (y=y_1,y_2). 
\label{step5}
\ena

In the tensor product
\be
Z=V_\z\otimes V_{q^{-1}\z}\otimes W_{q^{-1}\z},
\en 
there are two pairs which allow for non-trivial
$U_q\bp$-submodules: 
\be
&&V_0\subset V_{\z}\otimes V_{q^{-1}\z},
\\
&&
W_0\subset V_{q^{-1}\z}\otimes W_{q^{-1}\z}. 
\en
The submodule $V_0\simeq\C$ (resp. $W_0\simeq W_{q^{-2}\z}$) is 
spanned by $v_+\otimes v_-  -v_-\otimes v_+$ 
(resp. 
$\{v_-\otimes \ket{n}+(q^{2n}-1)v_+\otimes \ket{n-1}\}_{n\ge 0}$).
Set 
\be
Z_1=V_0\otimes W_{q^{-1}\z}, 
\quad 
Z_2=V_\z\otimes W_0. 
\en
Then a direct calculation shows that 
\be
&&y_1 Z\subset Z_1,\quad y_1 Z_1=0\,,
\\
&&y_2 Z\subset Z_1\oplus Z_2, \quad 
y_2 Z_1\subset Z_2,\quad y_2 Z_2\subset Z_1\,.
\en
Since $x Z_i\subset Z_i$ ($i=1,2$) 
holds for $x\in U_q\bp$, we have 
\be
\Tr_{c,a,A}\left(
y_i(\pi_\z\otimes\pi_{q^{-1}\z}\otimes
\varpi_{q^{-1}\z})(x)\right)=0
\quad (x\in U_q\bp)\,.
\en
The proof is now complete. 

\medskip

\noindent{\bf Step 4.}\quad

To complete the proof of Lemma \ref{prop:charact2} , 
we show matrix element of 
(\ref{step2}) between $\langle \kappa+\al|$, $|\kappa\rangle$
for $\mathbf{m=0}$.  
The integral consists of two parts, say $J_1$ and $J_2$, 
coming from the first 3 terms in \eqref{gb*0}, or from 
$\ub_{c,[1,m]}(\z,\al)$, respectively. 
We note that in view of \eqref{delinv-k} and Lemma \ref{1stAI} 
the proper meaning of $J_1$ is 
\be
&&J_1=\frac{1}{2}\int\limits_{\Gamma_{\bf 0}}\langle \kappa+\al|
\Tr _{[1,m],c}\Bigl(T_{[1,m],\mathbf{M}}(1,\kappa)T_{c,\mathbf{M}}(\z,\kappa)\kb_{ c,[1,m]}(\z,\al)(X_{[1,m]})\Bigr)|\kappa\rangle
\\
&&\quad\times
Q^-(\z,\kappa+\al)
\left(a(\z)Q^+(q^{-1}\z,\kappa)-d(\z)Q^+(q \z ,\kappa)\right)
\varphi(\z)
\frac{d\z^2}{\z^2}
\,.
\en
The functions 
\be
T_{c,\mathbf{M}}(\z,\kappa),\ \varphi(\z),\ \z^{\kappa+\al}Q^-(\z,\kappa+\al),\
\z^{-\kappa}Q^+(\z,\kappa), \
a(\z),\ d(\z) 
\en
are all regular at $\z^2=0$. 
On the other hand, for $X_{[1,m]}$ of spin $-1$, we have 
the estimate 
\be
&&\z ^{-\al}\kb_{[1,m]}(\z,\al)(X_{[1,m]})=O(\z^{2})
\quad (\z^2\to 0),
\en
as explained in \cite{HGSII}, section 2.5. The same argument there shows also  
\be
&&\z ^{-\al}\bl_{[1,m]}(\z,\al)(X_{[1,m]})
=O(\z^{2})\quad (\z^2\to 0).
\en
It follows that 
in both $J_1$ and $J_2$ the integrand is regular and  
the residue at $\z^2=0$ vanishes. 
This completes the proof of Lemma \ref{prop:charact2}.


\section{Classical limit}

In this Appendix, we explain the classical limit 
of our construction 
and its relation to hyperelliptic Riemann surfaces. 
We shall not go into much details since similar considerations
were done in \cite{SmiToda}, \cite{SmiKdV}. 
We assume $\al=0$.  
At the moment we are not ready to discuss the classical limit 
in the case $\al\neq 0$.   
We consider Bethe vectors of spin $0$, 
so that $\z^{\mp\kappa}Q^{\pm}(\z,\kappa)$ 
are polynomials in $\z ^2$ of the same degree
$$
s=\sum_{\mathbf{m=1}}^\mathbf{n} s_{\mathbf{m}}\,.
$$

In the parametrisation $q=e^{\pi i \nu}$, 
the classical limit amounts to $\nu\to 0$. 
So $\nu $ plays the role of  Planck's constant. 
Let us see what 
happens to the solutions to the Baxter equation in this limit. 
First of all, in order to obtain 
a Riemann surface of a finite genus,  
we keep $\mathbf{n}$ finite. 
But for the classical limit we 
need to have large quantum numbers. 
This is achieved by considering large spins $s_{\mathbf{m}}$. 
Actually, this was the main reason 
for us to consider arbitrary spins in the 
Matsubara direction.
So, we require that $\nu s_{\mathbf{m}}$, or equivalently $q^{s_\mathbf{m}}$,
tend to fixed non-zero values when $\nu\to 0$. 
Similarly, we demand that $\nu\kappa$, or $q^{\kappa}$, stays 
finite in the limit.
In this situation, $a(\z)$, $d(\z)$ and 
$T(\z)=T(\z, \kappa)$ tend to polynomials 
in $\z^2$ of degree $\mathbf{n}$, which 
we denote  by the same letters as in the quantum case.

In the classical limit, the poles of $\varphi(\z)$ concentrate 
on the portion of circles between the end points 
$\tau ^2_{\mathbf{m}}q^{-2s_{\mathbf{m}}}$ and 
$\tau^2 _{\mathbf{m}}q^{2s_{\mathbf{m}}}$.  
The $s$ zeros of the polynomials 
$\z ^{\mp\kappa}Q^{\pm}(\z,\kappa)$ 
concentrate to $\mathbf{n}$ open curved segments 
$C_{\mathbf{m}}$ close to the above circular segments.  
This claim is difficult to prove, but we can justify them
by analysing the Baxter equation in the classical limit. 

Consider the Baxter equation
\begin{align}
d(\z)Q(\z q)+a(\z)Q(\z q^{-1})=T(\z)Q(\z )\,.
\label{Ba}
\end{align}
We look for its quasi-classical solution in the form
\begin{align}
Q(\z)=
F(\z,\nu)\exp\Bigl
\{\frac 1 {2\pi i\nu }\int\limits _{1}^{\z^2}
\log \eta(\xi) \frac {d\xi ^2}{\xi^2} \Bigr\}, 
\label{qcl1}
\end{align}
where $\eta (\z)$ is a function independent of $\nu$ 
and $F(\z,\nu)$ is a power series in $\nu$.  
First, dividing the Baxter equation
by $Q(\z)$ 
and considering the leading order in Planck's constant
we conclude that $\eta (\z)$ must solve  the equation
\begin{align}
d(\z)\eta(\z) +a(\z)\eta^{-1}(\z)=T(\z)\,. 
\label{curve}
\end{align}
This is the equation of the 
spectral curve of the corresponding classical model. 

The function $\eta(\z)$ has two branches,  
$$
\eta^\pm(\z)=\frac {T(\z)\pm\sqrt{T(\z)^2-4a(\z)d(\z)}}{2d(\z)}
\,,
$$
for future convenience we choose the branch of the square root
such that 
$\sqrt{(q^{\kappa}-q^{-\kappa})^2}=q^{\kappa}-q^{-\kappa}$. 

Consider the behaviour $\z^2\to\infty$. 
The polynomial $T(\z)$ is not arbitrary, it comes from
the quasi-classical limit of a solution to the Baxter equation (\ref{Ba}).
Recall that for large $\z$ the function $Q^{\pm}(\z)$ is
$O(\z ^{\pm\kappa+2s})$
as discussed in Section 4. Also we have for $\z^2\to\infty$ and
to the main order
in  Planck's constant
\begin{align}
a(\z)&=\tau ^{-2} q^{2s}\z^{2\mathbf{n}}+\cdots,\quad 
d(\z)=\tau ^{-2}q^{-2s}\z^{2\mathbf{n}}+\cdots,
\nn\end{align}
where $\tau=\prod \tau _\mathbf{m}$.
So, the Baxter equation implies that
\begin{align}
T(\z)&=
\tau ^{-2}(q^{\kappa}+q^{-\kappa})\z^{2\mathbf{n}}+\cdots
\,.
\nn
\end{align}
Hence when  $\z^2\to\infty$ we have
$\eta^\pm\to q^{\pm\kappa+2s}$, which means
that $\eta ^+$ (resp. $\eta ^-$) corresponds to quasi-classical
limit of $Q^+$ (resp. $Q^-$).

Throughout this paper we use as parameter $\z$ while all the
important quantities are actually functions of $\z ^2$. This notational problem
is due to historical reasons, and we are forced to tolerate it in the
quantum case. However, in the classical case this notation
becomes very unnatural making incomprehensible simple formulae
for differentials on hyperelliptic Riemann surface. That is why in what follows
we shall often use the parameter $z=\z^2$. For the same reason we 
denote the discriminant $T(\z)^2-4a(\z)d(\z)$, which actually depends on $\z^2$,  by $P(\z^2)$.
Recalling that $a(\z)$, $d(\z)$, $T(\z)$ are polynomials of $\z ^2$ and making the
change of variables: $z=\z^2$, $w=2d(\z)\eta(\z)-T(\z)$ we bring the spectral curve (\ref{curve}) to 
canonical form:
\begin{align}
w^2=P(z)\,.\label{curve1}
\end{align}

In the $q$-deformed Abelian integrals the integration measure
contains $Q^{-}(\z ,\kappa)Q^{+}(\z ,\kappa)$. 
The most direct way to compute this quantity
uses the quantum Wronskian (\ref{wronskian})
\begin{align}
\frac{1}{q^\kappa-q^{-\kappa}}
W(\z) &=  Q^{+}(\z,\kappa )Q^{-}(\z q,\kappa)
-Q^{-}(\z ,\kappa)Q^{+}(\z q ,\kappa)
\nn\\
&\tto 
(\eta^--\eta^+)Q^{+}(\z ,\kappa)Q^{-}(\z,\kappa )\,.
\nn
\end{align}
This implies 
\begin{align}
Q^{+}(\z ,\kappa)Q^{-}(\z,\kappa )\varphi (\z)\ \tto 
\ \frac{1}{q^{-\kappa}-q^{\kappa}} \frac 1 {\sqrt{P(\z ^2 )}}\,,
\label{measure}
\end{align}
where we used the identity $W(\z)d(\z)\varphi(\z)=1$. 

The discriminant $P(z)$ is a polynomial 
of degree $2\mathbf{n}$. 
Let us call its zeros $x_{\bf 1},\cdots,x_{2\mathbf{n}}$.  
The Riemann surface (\ref{curve1}) 
is presented as two copies of the
$z$ plane glued together along 
the cuts $[x _{2\mathbf{m-1}},x_{2\mathbf{m}}]$.  
According to the conjecture accepted previously,  
the branch points can be ordered 
in such a way that the cut 
$[x _{2\mathbf{m-1}},x_{2\mathbf{m}}]$ 
is not far from the location of poles of 
$\varphi(\z)$ which
are contained in the contour $\Gamma _{\mathbf{m}}$. 
According to (\ref{measure}), 
for any polynomial $L(\z ^2)$ 
we have in the classical limit
\begin{align}
\int\limits_{\Gamma _{\mathbf{m}}}
L(\z ^2)Q^{+}(\z ,\kappa)Q^{-}(\z,\kappa )
\varphi (\z)\frac {d\z ^2}{\z ^2}\ 
\tto\
\frac{2}{q^{-\kappa}-q^{\kappa}}
\int\limits _{c _{\mathbf{m}}}
\frac {L(z) }{\sqrt{P(z )}}
\frac{dz }{z }\,,
\label{limint}
\end{align}
where $c _{\mathbf{m}}$ 
is a contour going in $z$-plane around 
$[x _{2\mathbf{m}-1},x_{2\mathbf{m}}]$ 
for $\mathbf{1}\le \mathbf{j}\le \mathbf{n}$, or
around $0$ for $\mathbf{m}=\mathbf{0}$. 
The limit (\ref{limint}) requires some remarks. 
The integral in the left hand side is
taken over the contour $\Gamma _{\mathbf{m}}$. 
In the limit 
the integrand develops cuts which appear 
as a result of
concentration of zeros of 
$Q^{+}(\z ,\kappa)Q^{-}(\z,\kappa )$
and poles of $\varphi (\z)$. 
So, obviously, in the limiting process 
we have to deform the contour in order that 
it does not cross the cut. 
This is how the integral around $c_\mathbf{m}$ appears.

The Riemann surface (\ref{curve1}) has genus $\mathbf{n-1}$.
The contours $c_{\mathbf{m}}$, with $\mathbf{m}=\mathbf{1},\cdots,
\mathbf {n-1}$ can be taken as $a$-cycles. 
Our Riemann surface have two points $0^\pm$ which lies on 
different sheets and project to $z=0$. The contour $c_\mathbf{0}$ 
goes around $0^+$. Similarly we have two points $\infty ^{\pm}$
which project to $z=\infty$.

Define the differentials on the Riemann surface
$$\sigma _{\mathbf{j}}(z)=\frac {z^{\mathbf{j}-1}}{\sqrt{P(z)}}dz,
\quad \mathbf{j}=\mathbf{0},\cdots, \mathbf{n}\,.$$
The differentials $\sigma _{\mathbf{j}}(z)$ where 
$\mathbf{j}=\mathbf{1},\cdots, \mathbf{n-1}$, are holomorphic
(the first kind) differentials while the differentials $\sigma _{\mathbf{0}}$ and $\sigma _{\mathbf{n}}$ are the third kind differentials. 
The differential $\sigma _{\mathbf{0}}$  has simple poles at
$z =0^\pm$, it is dual to the contour $c_\mathbf{0}$. 
The differential $\sigma _{\mathbf{n}}$ has simple poles at
$z =\infty^\pm$.

The holomorphic differentials can be normalised with respect 
to $c_{\mathbf{i}}$, $\mathbf{i}=\mathbf{1},\cdots,
\mathbf {n-1}$ because
$$
\det\Bigl( \ 
\int\limits_{c_{\mathbf{i}}}\sigma _{\mathbf{j}}
\ \Bigr)_{\mathbf{i,j=1,\cdots, n-1}}\ \ne \ 0\,.
$$
This is the classical version of (\ref{det}).

Consider the  differentials 
whose only singularities are at $\infty^{\pm}$.  
Among those are exact forms
\begin{align}
\frac {d}{dz}\(z^{k}\sqrt{P(z)}\)dz,\quad z^kdz,\quad k\ge 0\,.
\label{exactclassic}
\end{align}
Up to exact forms, holomorphic forms and the third kind differential $\sigma _{\mathbf{n}}(z)$
there are $\mathbf{n-1}$ linearly independent second kind differentials with singularities at  $\infty^{\pm}$:
$$
\tilde{\sigma }_{\mathbf{j}}(z)=z^\mathbf{j}\left[
\frac{d}{dz}
\bigl(z ^{-2{\mathbf{j}}}P(z)
\bigr)    \right]_+\frac {dz }{2\sqrt{P(z)}}\,, 
\quad \mathbf{j}=\mathbf{1,\cdots ,n-1}\,,
$$
where $[f(z)]_+$ means the polynomial 
part of $f(z )$,  
which is a Laurent polynomial at $z=\infty$.
We shall use at some point the differential $\tilde{\sigma}_\mathbf{0}$ which is an exact form.

The most important identity in the theory of Riemann
surfaces is the
Riemann bilinear relations. 
Usually this identity is written in the form
$$
\sum\limits _{\mathbf{m=1}}^{\mathbf{g}}
\Bigl(\  
\int
\limits_{a_{\mathbf{m}}}\omega _1
\int\limits _{b_{\mathbf{m}}}\omega _2
-
\int\limits_{b_{\mathbf{m}}}\omega _1
\int\limits _{a_{\mathbf{m}}}\omega _2\ \Bigr)
=2\pi i \ \omega _1\circ \omega _2
\,,
$$
where $\omega_{1,2}$ are the first or the second kind 
differentials, and  
\be
\omega _1\circ \omega _2
=-\sum \res (\omega _1d^{-1}\omega _2). 
\en
In our case the $a$-cycles coincide with 
$c_{\mathbf{1}},\cdots ,c_{\mathbf{n-1}}$.  
The $b$-cycle $b_{\mathbf{m}}$ ($\mathbf{m=1,\cdots ,n-1}$) 
crosses the cycle $a_{\mathbf{m}}$ once 
on the first sheet of the surface,
goes to the second sheet through the $\mathbf{m}$-th cut,
arrives to $\mathbf{n}$-th cut by the second sheet, 
crosses this cut and returns by the first sheet to its beginning.

An alternative way of writing the Riemann bilinear relations
is the following.  
It is easy to see that $\sigma$ 
and $\tilde{\sigma}$ constitute a canonical basis
\begin{align}
\sigma _{\mathbf{i}}\circ\tilde{\sigma} _{\mathbf{j}}=\delta _
{\mathbf{i},\mathbf{j}},\quad
\sigma _{\mathbf{i}}\circ\sigma _{\mathbf{j}}=0,\quad
\tilde{\sigma}_{\mathbf{i}}\circ\tilde{\sigma }_{\mathbf{j}}=0\,.
\label{intersect}
\end{align}
Now construct the antisymmetric form 
\begin{align}
\sigma (x,y)=
\sum\limits _{\mathbf{j=1}}^{\mathbf{n-1}}
\bigl(\sigma _{\mathbf{j}}(x)\tilde{\sigma} _{\mathbf{j}}(y)
-\sigma _{\mathbf{j}}(y)\tilde{\sigma} _{\mathbf{j}}(x)\bigr)
\,.
\label{sigma1}
\end{align}
Then
\begin{align}
\int\limits _{g_{\mathbf{1}}}\int\limits _{g_{\mathbf{2}}}
\sigma (x,y)=2\pi i \ g_{\mathbf{1}}\circ g_{\mathbf{2}}\,,
\label{ints}
\end{align}
where 
in the right hand side we put the intersection number of cycles.
From the explicit formulae 
for  $\sigma _{\mathbf{i}}$ and 
$\tilde{\sigma} _{\mathbf{j}}$, 
one easily finds the 2-form $\sigma (x,y)$,
\begin{align}
\sigma (x,y)=\Bigl(\frac{\partial}{\partial y}\Bigl(\frac{1}{y-x}\frac {\sqrt{P(y)}} {\sqrt{P(x)}}\Bigr)
-\frac{\partial}{\partial x}\Bigl(\frac{1}{x-y}\frac {\sqrt{P(x)}} {\sqrt{P(y)}}\Bigr)
\Bigr)dxdy\,.
\label{sigma2}
\end{align}
This form is exact, so, apparently
the integrals over all 2-cycles must vanish. 
However, there is a singularity at $x=y$
which produces the intersection number in the right
hand side of (\ref{ints}). 
All that is quite standard, so we do not go into
much details. 

Consider a particular case of (\ref{ints}),
\begin{align}
\int\limits _{c_{\mathbf{i}}}\int\limits _{c_{\mathbf{j}}}
\sigma (x,y)=0\,,\quad \mathbf{i},\mathbf{j=
1,\cdots ,n-1}\,.\label{rieclassic}
\end{align}
This is true because the $a$-cycles do not intersect.

On the product of two copies of Riemann surface we
have the canonical second kind differential
$\rho (x,y)$ with the following properties. 
\begin{itemize}
\item The differential $\rho (x,y)$ is holomorphic everywhere
except the diagonal, where it has a double pole with no residue
\begin{align}
\rho (x,y)=\(\frac 1 {(x-y)^2}+O(1)\)dx dy\,.
\label{dpole}
\end{align}
\item
The differential $\rho (x,y)$ is normalised
with respect to $x$, 
\begin{align}
\int\limits _{c_{\mathbf{m}}}
\rho (x,y)=0,\qquad\mathbf{m}
=\mathbf{1},\cdots ,\mathbf{n-1}\,.
\label{normrho}
\end{align}
\end{itemize}
An important consequence of the Riemann bilinear relations
is that this differential is automatically symmetric:
\begin{align}
\rho (x,y)=\rho (y,x)\,.
\label{symrho}
\end{align}
Let us explain this by 
giving an explicit construction of $\rho (x,y)$. 
We start with an exact form in $x$, 
$$
-\frac {\partial}{\partial x}
\(\frac {\sqrt{P(x)}}{\sqrt{P(y)}(x-y)}\)dx dy\,.
$$
which obviously has the required singularity at $x=y$, 
but 
has also additional singularities at infinity. 
Because of \eqref{sigma1} and \eqref{sigma2}, these singularities are 
cancelled in the following expression:
$$
\rho(x,y)=-\frac {\partial}{\partial x}
\(\frac {\sqrt{P(x)}}{\sqrt{P(y)}(x-y)}\)dx dy
+\sum_{\mathbf{i}=1}^{\mathbf{n}-1} 
\tilde{\sigma}_{\mathbf{i}}(x)\sigma _{\mathbf{i}}(y)+
\sum_{\mathbf{i},\mathbf{j=1}}^{\mathbf{n}-1} 
 X_{\mathbf{i},\mathbf{j}}\sigma _{\mathbf{j}}(x)
\sigma _{\mathbf{i}}(y)\,,
$$
where the matrix $X_{\mathbf{i},\mathbf{j}}$ 
must be defined from the normalisation condition
$$
\sum\limits_{\mathbf{j}=1}^{\mathbf{n}-1}X_{\mathbf{i},\mathbf{j}}
\int\limits _{c_{\mathbf{k}}}\sigma _{\mathbf{j}}
+\int\limits _{c_{\mathbf{k}}} \tilde{\sigma}_{\mathbf{i}}
=0\,.
$$
Now writing a similar formula for $\rho(y,x)$, it becomes apparent that 
the symmetry (\ref{symrho}) is equivalent to
the fact that $X$ is a symmetric matrix. 
This fact follows from (\ref{rieclassic}). 
There is an obvious similarity 
between this argument and the proof of Lemma \ref{symmetry}.

Suppose that we want to construct a normalised 
second kind differential
with given singular part. 
To be more precise, we allow a singularity
only at $x=1$ with a given singular part
$$
\tau _{sing}(x)=\sum\limits_{k=2}^N\gamma_k(x-1)^{-k}dx\,.
$$
So, we look for a differential which has the 
singular part $\tau _{sing}(x)$ at $x=1$ and is holomorphic
elsewhere.
We require that $\tau(x)$ is normalized
$$
\quad \int\limits _{c_{\mathbf{m}}}\tau(x)=0\,.
$$
It is rather obvious that $\tau(x)$ is given by
\begin{align}
\tau(x)=\oint\limits _{\Gamma} \sigma (x,y)\ 
d^{-1}\tau _{sing}(y)\,,
\label{2ndkind}
\end{align}
where the contour $\Gamma$ is as usual: 
$1$ is inside it, and $x$ outside.

Let us return to the 
quasi-classical limit of the quantum formulae.
First, notice that  for $\al=0$ 
the operator $\overline{D}$ becomes the
second difference derivative because $\rho(\z)=1$:
$$
\frac 1 {(\pi i \nu)^2}\overline{D}_\z (f(\z))=
\frac 1 {(\pi i \nu)^2}\(f(\z q)+f(\z q^{-1})-2f(\z)\)\tto
\( \z\frac{d}{d\z} \)^2f(\z)\,.
$$
Also $\Delta_\z^{-1}$ goes to the primitive function
$$
2\pi i\nu\Delta _\z^{-1}(f(\z))\tto 
\(\z\frac{d}{d\z} \)^{-1}f(\z)\,.
$$
Consider the $f(\z)=L(\z^2)$
and the corresponding $q$-deformed exact form (for $\al=0$
there is no difference between $f^{\pm}(\z)$):
\begin{align}
\varpi _\nu(\z^2)=\frac 1 {\pi i \nu}E (f(\z))Q^-(\z )Q^+(\z)\varphi (\z)\frac{d\z ^2}{\z ^2}\nn 
\end{align}
then
\begin{align}
\varpi _\nu(z)\ \tto 
\ - \frac{d}{dz }\(L(z)\sqrt{P(z)}\)dz \,.
\nn
\end{align}
Denote 
$$
\sigma _{\nu}(\z ^2,\xi ^2)=\frac 1 {\pi i \nu}r(\z,\xi)Q^-(\z )Q^+(\z)\varphi (\z)
Q^-(\xi)Q^+(\xi)\varphi (\xi)\frac{d\z ^2}{\z ^2}\frac{d\xi ^2}{\xi ^2}
\,,
$$
Then we have
\begin{align}
\sigma _{\nu}(x,y)\ \tto\  \sigma(x,y)+\frac 1 2 \( \sigma _\mathbf{0}(x)\tilde{\sigma} _\mathbf{0} 
(y)
-\sigma _\mathbf{0}(y)\tilde{\sigma} _\mathbf{0}(x) \)\,,\label{rielimit}
\end{align}
the additional term is not important in \eqref{ints} because $\tilde{\sigma} _\mathbf{0}$ is an exact form.
The limit  (\ref{rielimit})
explains the name $q$-deformed Riemann bilinear relations
for (\ref{riemann}).

Consider now 
$$
\rho _{\nu}(\z^2,\xi^2)=\frac 1 {\pi i\nu}
T(\z)T(\xi)\omega(\z,\xi)Q^-(\z )Q^+(\z)\varphi (\z)
Q^-(\xi)Q^+(\xi)\varphi (\xi)\frac{d\z ^2}{\z ^2}\frac{d\xi ^2}{\xi ^2}\,.
$$
We want to show that
$$\rho _{\nu}(x,y)\ \tto\ \rho (x,y)\,.$$
First, it is rather easy to find that in the singularity (\ref{omegasing})
two simple poles produce in the classical limit 
the 
double pole in (\ref{dpole}). 
Second, we have the normalisation conditions
(\ref{normofom}). They look different from the normalisation
conditions (\ref{normrho}) because of presence of the
term 
\begin{align}
&\frac 1 {\pi i\nu}\int\limits _{\Gamma _{\mathbf{m}}}T(\z,\kappa)\overline{D}_{\z}\overline{D}_{\xi}\Delta _{\z}^{-1}\psi (\z/\xi)
Q^-(\z )Q^+(\z)\varphi (\z)\frac{d\z ^2}{\z ^2}\label{bad}\,.
\end{align}
However, this term for $\al=0$, $\nu\to 0$ is of order
$\nu ^2$, while $\rho _{\nu}(\z ^2,\xi ^2)$ is of order $1$. So the
term (\ref{bad}) does not count and from (\ref{normofom})
with $\mathbf{m=1,\cdots ,n-1}$ we get the normalisation
conditions (\ref{normrho}).
Conditions (\ref{normofom}) 
with $\mathbf{m=0,n}$ show that the differential
$\rho _{\nu}(\z ^2,\xi ^2)$ in the limit $\nu\to 0$ does not have 
simple poles at $\z^2=0,\infty$ which were originally present.

Thus we conclude that the function 
$\omega(\z,\xi)$ is related
 in the classical limit to the canonical normalized second kind
differential.

Notice a clear similarity between the formula (\ref{2ndkind}) and
our main formula (\ref{b*to c1}).

\section{Equivalence of different non-degeneracy conditions.}
\label{nondeg}

In this Appendix we show that the conditions 
$\det(\mathcal{A}^{\pm})\ne 0$  are equivalent to the fact that the scalar product
\eqref{general} does not vanish. We use usual notations of the 
Quantum Inverse Scattering Method (QISM) \cite{FST}:
$$T_{a,\mathbf{M}}(\z)=\begin{pmatrix}
A(\z)&B(\z)\\C(\z)&D(\z)
\end{pmatrix}_a\,.$$
Consider the case when all the spaces in Matsubara direction
are two-dimensional (spin $1/2$). The basis of the two-dimensional space
will be denoted by $e_{\pm}$. Introduce two vectors
in Matsubara space
\begin{align}
|+\rangle =e_{+}\otimes\cdots\otimes e_{ +},\ ,
\quad |-\rangle =e_{-}\otimes\cdots\otimes e_{ -}\,.
\label{natbase}
\end{align}
The eigenvector $|\kappa\rangle $ is written in QISM framework as
\begin{align}
|\kappa\rangle=\prod C(\la ^-_\bj)|-\rangle\,,
\label{BBB}
\end{align}
where $(\la ^-_{\bj}) ^2$ are zeros of
$\z^{\kappa}Q^-_{\mathbf{M}}(\z,\kappa)$
which is a polynomial of $\z^2$. It is well-known that this eigenvector
does not vanish identically unless $\tau _\bi=\tau _\bj q$ for some
$\bj>\bi$. The latter situation has to be forbidden from the very
beginning because the tensor product on $\bi$-th and $\bj$-th
spaces is reducible and contains one-dimensional sub-module.
On the other hand there is no problem with the case 
$\tau _\bi=\tau _\bj q^{-1}$ 
which allows the fusion procedure, and show that
our considering only spin $1/2$ representations is not a real restriction. 

Consider now the vector $\prod B(\la ^+_\bj)|+\rangle$, where
$(\la ^+_{\bj}) ^2$ are zeros 
of $\z^{-\kappa}Q^+_{\mathbf{M},\kappa}(\z)$.
This vector also does not vanish identically, it is an eigenvector
of $T_{\mathbf{M}}(\z,\kappa )$ with the same eigenvalue as
\eqref{BBB}. Hence, the assumed uniqueness
of the eigenvector with the eigenvalue of 
maximal absolute value implies that this
vector is proportional to $|\kappa\rangle$ with some coefficient
which depends on $\tau_{j}$ and $\kappa$, the exact form of this coefficient
is irrelevant here. 

Now consider the scalar product \eqref{general}. We do not care about
the normalisation of the eigenvectors, so, in traditional QISM
way it is written as
$$\langle\kappa+\al|\kappa\rangle=
\langle -|\prod B(\mu ^-_\bj)\prod C(\la ^+_\bj)|-\rangle\,,$$
where $(\mu _\bj^{\pm})^2$ are zeros of $\z ^{\mp
\kappa}Q^\pm_{\mathbf{M}}(\z,\kappa)$.
Due to the previous remark we rewrite:
\begin{align}
\langle\kappa+\al|\kappa\rangle=Const\cdot
\langle -|\prod B(\mu ^-_\bj)\prod B(\la ^+_\bj)|+\rangle\,,
\label{BBBBBB}
\end{align}
where 
$Const$ is a
nonvanishing constant
which was discussed above. So, we conclude that the scalar
product in question is given essentially by the partition
function with domain wall boundary conditions
$$M_{\bn}(\xi _{\bo},
\cdots ,\xi_{\bn}|\tau _{\bo},
\cdots ,\tau_{\bn})
=\prod\xi_{\bj}^{-1}
\langle -|\ \prod\limits _{\bj=1}^{\bn} B(\xi _\bj)\ |+\rangle$$
with specification $\{\xi _\bj\}=\{\mu ^-_\bj\}\cup\{\la ^+_\bj\}$, notice that
independently of spin of our eigenvectors the number of
elements in the latter set is $\bn$.

Being a polynomial
of degree $\mathbf{n-1}$ in $\xi _{\bn}^2$ the function $M_{\bn}$ is completely characterised by the recurrence
relation:
\begin{align}
&M_{\bn}(\xi _{\bo},
\cdots ,\xi_{\bno},\tau _\bn|\tau_{\bo},
\cdots \tau_{\bno},\tau_{\bn})\label{rec1}\\
&= (q^2-1)\tau_\bn^{-1}\prod\tau^{-2}_\bj\prod\limits_ {\bj\ne \bn}
(q^2\xi _{\bj}^2-\tau _\bn^2)(q^2\tau_\bn^2-\tau _\bj^2)
M_{\bno}(\xi _{\bo},
\cdots ,\xi_{\bno}|\tau_{\bo},
\cdots \tau_{\bno})\,.\nn
\end{align}
This recurrence was solved by Izergin who found a determinant
formula for $M_{\bno}$ \cite{izergin}.

On the other hand we have the determinant $\det (\mathcal{A}^+)$
of $(\mathbf{n+1})\times (\mathbf{n+1})$ matrix. 
This determinant depends on the Bethe roots only through
the product $Q^-(\z,\kappa+\al)Q^+(\z,\kappa)$. Once again we
consider the union $\{\xi _\bj\}=\{\mu ^-_\bj\}\cup\{\la ^+_\bj\}$ and
normalise this product as follows
$$Q^-(\z,\kappa+\al)Q^+(\z,\kappa)=\prod\limits _{\bj=1}^\bn
(\z^2-\xi _\bj^2)\,.$$
The determinant can be reduced in two steps:
\begin{align}
&\det (\mathcal{A}_{\bi,\bj}^+)_{\mathbf{i,j=0,\cdots, n}}=-2\pi i
\prod \xi ^2_\bj\det (\mathcal{A}_{\bi,\bj}^+)_{\mathbf{i,j=1,\cdots, n}},\label{n+1nn-1}\\
&\det (\mathcal{A}_{\bi,\bj}^+)_{\mathbf{i,j=1,\cdots, n}}=-2\pi i
\det (\mathcal{A}_{\bi,\bj}^+)_{\mathbf{i,j=1,\cdots, n-1}}\,,\nn
\end{align}
where we used the obvious identities:
\begin{align}
&\int\limits _{\Gamma _{\mathbf{0}}}
\z ^{\al +2\mathbf{j}}
Q^-(\z ,\kappa+\al)Q^+(\z,\kappa)\varphi (\z)\frac{d\z ^2}{\z ^2}=
(-1)^{\bn-1}2\pi i\delta _{\mathbf{j},\mathbf{0}}\
\prod \xi ^2_\bj\,,
\nn\\
&\int\limits _{\Gamma _{\mathbf{\infty}}}
\z ^{\al +2\mathbf{j}}
Q^-(\z ,\kappa+\al)Q^+(\z,\kappa)\varphi (\z)\frac{d\z ^2}{\z ^2}=
-2\pi i\delta _{\mathbf{j},\mathbf{n}}\,,  
\nn
\end{align}

Making the dependence on $\bn$ and
other parameters explicit we introduce
\begin{align}&D_{\bn}(\xi _\bo,\cdots, \xi_\bn
|\tau _{\bo},
\cdots ,\tau_{\bn})
\nn\\ &\quad=(-1)^{\bn(\mathbf{n-1})/2}
\prod\tau _\bj^{-2}\ 
\prod\limits _{\bi,\bj}(
q\tau ^2_\bi-q^{-1}\tau ^2_\bj)
\prod\limits_{\bi<\bj}(\tau _\bi^2-\tau _\bj^2)
\det (\mathcal{A}^+_{\bi,\bj})_{\mathbf{i,j=1,\cdots n}}\,.\nn
\end{align}
where we preferred the intermediate reduction from
\eqref{n+1nn-1} for its antisymmetry with respect to
permutation of $\tau $'s.
In the case of two-dimensional representations in Matsubara
direction the integrals in $\mathcal{A}^+_{\bi,\bj}$ are easy:
they are given by sum of two residues. 
Obviously, 
$D_{\bn}$ is a polynomial in $\xi ^2_\bn$ of degree $\bn$.
However the second relation from \eqref{n+1nn-1} shows that the
actual degree is $\bno$.

Set $\xi _\bn=\tau _\bn$ and multiply the matrix
$\mathcal{A}^+$ from the right by the
matrix $I-\tau^2_\bn E$ with $E_{i,j}=\delta _{i,j-1}$. Then it is easy to see that in the last row only $\bn$-th matrix element does not
vanish.
Using this,
after some simple algebra one sees that $D_\bn$ satisfies
the relation \eqref{rec1}.  Hence
we conclude that
$$D_{\bn}(\xi _\bo,\cdots, \xi_\bn
|\tau _{\bo},
\cdots ,\tau_{\bn})=M_{\bn}(\xi _\bo,\cdots, \xi_\bn
|\tau _{\bo},
\cdots ,\tau_{\bn})\,.$$
Due to the above reasoning it shows that
$\langle \kappa +\al|\kappa\rangle $
is proportional to $\det (\mathcal{A}^+)$ with
non-vanishing coefficient. Similarly,  rewriting 
$\langle \kappa +\al|\kappa\rangle $ as 
$$\langle\kappa+\al|\kappa\rangle =Const\cdot
\langle +|\prod C(\mu ^+_\bj)\prod C(\la ^+_\bj)|-\rangle\,,$$
one proves that it is proportional to $\det (\mathcal{A}^-)$
with
non-vanishing coefficient.

\bigskip

\noindent

{\it Acknowledgements.}\quad
Research of MJ is supported by the Grant-in-Aid for Scientific 
Research B-20340027 and B-20340011. 
Research of TM is supported by
the Grant-in-Aid for Scientific Research B--17340038.
Research of FS is supported by  EC networks   "ENIGMA",
contract number MRTN-CT-2004-5652
and GIMP program (ANR), contract number ANR-05-BLAN-0029-01.

The authors are grateful to O. Babelon,
F. G{\"o}hmann, N. Kitanine, S. Lukyanov and L. Takhtajan
for helpful discussions. Special thanks are due to H. Boos for
long and fruitful collaboration which we hope to continue.
TM and MJ wish to thank
Universit{\'e} Paris VI for kind hospitality where this work began.

\end{document}